\documentclass[11pt]{llncs}
\usepackage[margin=1in]{geometry}
\usepackage[final]{graphicx}

\usepackage{times}
\usepackage{graphicx,psfrag,url}
\usepackage{subfigure}
\usepackage{ifpdf}
\usepackage{color}
\usepackage{amsmath}
\usepackage[makeroom]{cancel}
\usepackage{amssymb,verbatim,ifthen}
\usepackage{amsfonts}
\usepackage{graphicx}
\usepackage{enumerate}
\usepackage{mathrsfs}
\usepackage{bbm}
\usepackage{hyperref}
\usepackage{mathrsfs}

\usepackage{refcount}

\usepackage{natbib}

\usepackage{soul,xcolor}

\definecolor{purple}{rgb}{1, 0, 0.4}
\definecolor{dgreen}{rgb}{0, 0.4, 0}

\newcommand{\thl}[1]{\textcolor{mygreen}{#1}}

\newcommand{\hide}[1]{}

\ifpdf
  \newcommand{\frag}[2]{}
  \newcommand{\fragr}[2]{}
  \DeclareGraphicsExtensions{.jpg,.pdf,.mps,.png}
\else
  \usepackage{graphicx,psfrag}
  \newcommand{\frag}[2]{\psfrag{#1}[cm][cm]{\small #2}}
  \newcommand{\fragr}[2]{\psfrag{#1}[br][br]{\small #2}}
  \DeclareGraphicsExtensions{.eps}
\fi

\definecolor{myred}{rgb}{0.95, 0.01, 0.01}
\definecolor{mygreen}{rgb}{0.2, 0.55, 0.4}
\definecolor{myblue}{rgb}{0.01, 0.01, 0.9}

\newcommand{\todo}[1]{}

\newcommand{\R}{\mathbb{R}}

\newcommand{\F}{\mathcal{F}}

\newcommand{\pathcost}{c}
\newcommand{\paths}{\mathcal{P}}
\newcommand{\route}{\pi}

\newcommand{\std}{ \sigma }
\newcommand{\stot}{s\mbox{--}t}
\newcommand{\dsh}{\mbox{--}}

\newcommand{\stdev}{deviation}
\newcommand{\ntrl}{single-minded}
\newcommand{\divpar}{diversity parameter}
\newcommand{\G}{\mathcal{G}}

\date{}

\begin{document}

\title{When Does Diversity of User Preferences Improve Outcomes in Selfish Routing?\thanks{Part of this work was completed while the authors were visiting the Simons Institute for the Theory of Computing, Berkeley, CA. Research partially supported by  NSF grants CCF-1527568, CCF-1216103, CCF-1350823, CCF-1331863, CCF 1733832.} }

\author{Richard Cole\inst{1} \and Thanasis Lianeas\inst{2}
\and Evdokia Nikolova\inst{2}}

\institute{New York University
\and
University of Texas at Austin
}

\maketitle

\begin{abstract}
We seek to understand when heterogeneity in user preferences yields improved outcomes in terms of overall cost. That this might be hoped for is based on the common belief that diversity is advantageous in many settings. We investigate this in the context of routing. Our main result is a sharp characterization of the network settings in which diversity always helps, versus those in which it is sometimes harmful.

Specifically, we consider routing games, where diversity arises in the way that users trade-off two criteria (such as time and money, or, in the case of stochastic delays, expectation and variance of delay). 
Our main contributions---a conceptual and a technical one--- are the following:\\
1) A participant-oriented measure of cost in the presence of user diversity, together with the identification of the natural benchmark: the same cost measure for an appropriately defined average of the diversity.\\
2) A full characterization of those network topologies for which diversity always helps, for all latency functions and demands. For single-commodity routings, these are series-parallel graphs, while for multi-commodity routings, they are the newly-defined ``block-matching'' networks. The latter comprise a suitable interweaving of multiple series-parallel graphs each connecting {a} distinct 
source-sink {pair}.

While the result for the single-commodity case may seem intuitive in light of the well-known Braess paradox, the two problems are different: there are instances where diversity helps although the Braess paradox occurs, and vice-versa. But the main technical challenge is to establish the ``only if'' direction of the result for multi-commodity networks. This follows by constructing an instance where diversity hurts, and showing how to embed it in any network which is not block-matching, by carefully exploiting the way the simple source-sink paths of the commodities intersect in the ``non-block-matching'' portion of the network.

\end{abstract}

\setcounter{footnote}{0}

\section{Introduction}

People are inherently diverse and it is a common belief that diversity helps.  In one of the central  themes of algorithmic game theory---the tension between selfish behavior and social optimality---can diversity of user preferences help to bring us closer to the coveted social optimality? 
We provide an answer to this question in the context of non-atomic selfish routing, where diversity naturally arises in the way users trade-off two criteria, for example, time and money, or, in the presence of uncertain delays, expectation and variance of delay.  

Diversity is reflected in combining the two criteria via different individual coefficients, which we call the `\divpar{s}'.  
We consider a linear combination of the two criteria, as in the literature on tolls  where users minimize travel time plus tolls, e.g.,~\citep{BeMcWi56,FlJaMa04} or the literature on risk-averse selfish routing where users minimize expected travel time plus variance (\citep{NSM11}), or more generally travel time plus a deviation function (\citep{DBLP:conf/sagt/KleerS16}).  

We are interested in understanding whether
heterogeneity in user preferences improves collective outcomes or makes them worse.  As we shall see, there is no unique answer. Rather, it depends on the setting.
To address our question we need to specify how to measure the cost
of an outcome, and define our comparison point, namely a benchmark
setting with no user heterogeneity.  As explained above, 
to measure the cost of an outcome, we treat a user's cost as the 
sum of two terms associated with two criteria:
If we let $\ell_P$ denote the cost of one criterion (e.g., the latency) over a path $P$, 
and $v_P$ be the cost of the second criterion, then the overall cost is given by  $\ell_P + r\cdot v_P$, 
where $r$ is our \divpar.
The special case of $r=0$ corresponds to indifference to the second criterion and results in the classic selfish routing model where users simply minimize travel time.

A first approach to measuring the effect of diversity might be to compare the cost of an outcome with $r=0$  (i.e., just 
the total latency) to that with other values of $r$, including possibly mixed values of $r$ across the population being routed.
However, this approach does not pinpoint the gains and losses from user 
heterogeneity as opposed to user homogeneity;
rather, it (mostly) pinpoints the gains and losses depending
on whether players are affected by the second criterion or
not.
Instead, we focus on the sum of the costs incurred by the users
as measured by their cost functions, and compare costs
incurred by a heterogeneous population of users to those incurred
by an equivalent population of homogeneous users.
What are equivalent populations?
Suppose the heterogeneous population's diversity profile is given by 
a population density function $f(r)$.
Then, we define the corresponding homogeneous population to 
have {the single} diversity parameter $\bar r = \int r f(r) dr$.
In addition, we require the two populations to have the same size, in the 
sense that the total source-to-sink flows that they induce {are} equal.

For this setting, we completely characterize the graphs for which user heterogeneity 
does no harm, in the sense of only reducing
the total cost as perceived by users, for both the case of a single  
commodity (i.e.\ one source-to-sink flow) and for multiple
commodities (i.e.\ {flows between} multiple {source-sink} {pairs}).
For the single-commodity case, these graphs are exactly the series-parallel
graphs, with the source being the ``start'' vertex of the graph, and the
sink being the ``terminal'' vertex.
For the multi-commodity case, each commodity flows over a series-parallel graph, 
and further these subgraphs need to overlap in a particular
manner which we specify later, in the formal statement of results.
For all other graphs, we provide examples of desired source-to-sink demands
for which the resulting equilibrium flows are more expensive than the flows
in the corresponding homogeneous problem.

\smallskip \noindent \textbf{Related work.} 
To the best of our knowledge this is the first work that methodically  compares the effects of heterogeneity and homogeneity in user preferences for a self-interested 
routing population. In fact, in the broader area of algorithmic game theory, this seems to be the first time that a question of this type has been considered, with the exception of \cite{chenetalaltruism}. \footnote{In a different setting, \cite{chenetalaltruism} show how diversity may affect a bound they prove for the price of anarchy on parallel link networks.}  Algorithmic game theory research mentioning diversity exists in the context of the theory of evolution (e.g., \cite{DBLP:conf/innovations/MehtaPP15,DBLP:conf/innovations/ChastainLPV13}), which is very different from our focus.   

Since we provide attitudes to time versus money and time versus risk as motivating examples for user diversity, we briefly mention related work on tolls and on risk-averse selfish routing.  Regarding the latter, 
there are multiple ways to model how the behavior of players incorporates risk when uncertainty is present (see e.g.~\cite{Rock07}). \citet{PNS13} studied the effect that different risk attitudes may have on a system's performance at equilibrium. They did so by examining the price of anarchy (i.e.\ the ratio of the cost at equilibrium to the optimal cost) for different risk formulations.
\citet{NSM15} and \citet{DBLP:journals/corr/LianeasNM15}
studied the degradation of a network's performance due to risk aversion. This kind of degradation is captured by the price of risk aversion, which compares the cost of the equilibrium when players are risk-{averse} to the equilibrium cost when players are risk-{neutral}. The above works assumed that all players have the same risk averse preferences, which we call risk homogeneity; they do not offer any indication as to whether and under what circumstances risk heterogeneity improves or harms a system's performance.  In contrast, \citet{DBLP:conf/wine/FotakisKL15} considered games with heterogeneous risk-averse players.  They showed how uncertainty may and can be used to improve a network's performance, but the effects of heterogeneity as opposed to homogeneity were left unexamined.
Regarding the related literature on tolls, early results (e.g., \cite{BeMcWi56}) 
showed that tolls can help 
implement the social optimum  as an equilibrium, when users 
all have the same linear objective function combining time and money.
 Much more recently, 
these results were extended to the case where users trade-off travel time and money differently, by \citet{DBLP:conf/stoc/ColeDR03} and \citet{Fleischer2005217} for the single commodity case, and by \citet{KaKo04} and \citet{FlJaMa04} for the multicommodity case. 
We remark that in the above works,  apart from \citet{DBLP:conf/stoc/ColeDR03}, the social welfare is defined as the total travel time, whereas in our work we consider the total user cost, which encapsulates both criteria. This, for example, is also the case for \citet{DBLP:journals/algorithmica/ChristodoulouMP14} and \citet{DBLP:conf/caan/KarakostasK04}. We further note that if 
the social welfare is defined as the total travel time then there are simple instances on parallel link graphs where diversity is harmful.

Characterizing the topology of networks that satisfy some property 
is a common theme in computer science. 
Relevant to our work, \citet{Epstein2009115}  characterized the topology of single-commodity networks for which all Nash equilibria are social optima (under bottleneck costs), and  \citet{DBLP:journals/geb/Milchtaich06}  characterized the topology of single-commodity networks which do not suffer from the Braess Paradox for any cost functions. 
 \citet{DBLP:conf/sagt/ChenDH15}  fully characterized the topology of multi-commodity networks that do not suffer from the Braess paradox. These characterizations appear similar to ours, although there does not seem to be any other connection between the two problems, as {(i)} there are instances where diversity helps {while} the Braess Paradox occurs and others where diversity hurts but the paradox does not occur, and {(ii)} the Braess Paradox may occur in series-parallel networks when considering selfish routing with heterogeneity in 
 user preferences, which is not the case for the classic selfish routing model. 

In the following works,  the characterizing
topology for the corresponding question (for a single commodity) is similar to ours. \citet{DBLP:journals/im/FotakisS08}  considered atomic games and proved that series-parallel networks are the largest class of networks for which strongly optimal tolls are known to exist. \citet{NSM15}  considered
homogeneous agents and a social cost function that does not account for the second criterion; They showed that series-parallel networks admit the best bound on the degradation of the network due to risk aversion. Theorem 4 of \citet{DBLP:journals/corr/AcemogluMMO16},  proves that series-parallel
networks are the characterizing topology for what they call the Informational Braess Paradox with Restricted Information Sets;
this theorem compares the cost of  one agent type before and after more information is revealed to agents of that type, but does not consider the change in the cost of other agent types.
In contrast, our work considers non-atomic games with heterogeneous agents and bounds the overall costs
faced by the collection of agents. Most relevant to our work is \cite{DBLP:journals/corr/MeirP14}  and its Theorem
3.1 as it implies that for series-parallel networks the cost of
an agent of average parameter only increases when switching
from the heterogeneous instance to the corresponding homogeneous
one and thus for our sufficiency theorems, one is left
to prove that the heterogeneous equilibrium cost is no greater
than the cost of an agent of average parameter (though we give a different proof).

\smallskip
\noindent \textbf{Contribution.}
We fully characterize the topology of networks for which diversity is never harmful, regardless of the demand size and the distribution of the \divpar{ }(discrete or continuous). 
We do so both for single and multi-commodity networks. 
%

For single-commodity networks it turns out that this topology is that of series-parallel networks.  In Theorem~\ref{thm:SePaSufficient}, we show that if the network is series-parallel, then diversity
only helps for any choice of demand and edge functions. 
The key observation is that 
there is a path for which the homogeneous flow is at least as large as the heterogeneous flow.
As the cost of the homogeneous flow is the same on all used paths, while the cost of each unit of heterogeneous flow is lowest on the path it uses, one can then deduce that the cost of the heterogeneous flow
is at most that of the homogeneous flow. To show necessity,
we first provide an instance on the Braess graph for which diversity is harmful,
and then show how to embed it in any non-series-parallel graph.

In multi-commodity networks, by the result 
above,  each commodity must route its flow through a series-parallel subnetwork. But, as Proposition \ref{prop:JustTwoPaths} shows, this is not enough, and the way in which these series-parallel networks overlap needs to be constrained. The necessary constraint is exactly captured by the class of \emph{block-matching} networks, 
defined in this paper. 
Sufficiency in this case then follows quite easily from the same result for the single commodity case.

The main technical challenge is to show necessity.
To this end, assuming 
diversity does no harm, we show, via a case analysis, how the subnetworks of the commodities may overlap.
First, in Proposition \ref{prop:JustTwoPaths}, we give an instance on a network of two commodities and three paths for which diversity hurts. 
Then we mimic this instance on a general network.
The  difficult part is to choose the
corresponding paths for the mimicking, so that, in the created instance, all the flow under both equilibria goes through these paths. 
The challenge 
is that the commodities' subnetworks may overlap in subtle ways.

\section{Preliminaries}
We consider a directed multi-commodity network $G=(V,E)$ with an aggregate demand of $d_k$
units of flow between origin-destination pairs $(s_k,t_k)$ for $k\in K$. We
let $\paths_k$ be the set of all paths between $s_k$ and $t_k$, and
$\paths:=\cup_{k\in K} \paths_k$ be the set of all origin-destination paths.  
We let $[m]$ denote $\{1,\ldots, m\}$.  We assume that
 $K=[m]$, for some $m$.  The users in the
network---i.e., the players of the game---must choose routes that connect
their origins to their destinations. We encode the collective decisions of
users in a flow vector $f =(f_{\route})_{\route\in\paths}\in
\R^{|\paths|}_+$ over all paths. Such a flow is feasible when demands are
satisfied, as given by constraints $\sum_{\route\in\paths_k} f_{\route}=d_k$
for all $k\in K$. For simplicity, we let $f_e$ denote the flow on edge $e$;
note that $f_e = \sum_{\route :e\in \route}
f_{\route}$. When we need multiple flow variables, we use the analogous
notation $g, g_{\route}, g_e$.

The network is subject to congestion that affects 
two criteria  the players consider.  
These two criteria are modeled by two \emph{edge-dependent functions} that take as input the edge flow $f_e$ of $e$, for each edge $e$: a latency function $\ell_e(x)$ assumed to be continuous and non-decreasing,  and a \stdev{ }function $\std_e(x)$
  assumed to be continuous (but not necessarily non-decreasing).
  Function $\ell_e(\cdot)$ represents the first criterion while  $\std_e(\cdot)$ represents the second criterion.

Players choose paths according
to a linear combination of 
the
first criterion  and the 
second criterion   along the route. Throughout the paper we refer to the players'
objective as the cost along a route. 
Formally, for a given user, on
letting $\ell_{\route}(f) =\sum_{e\in\route} {\ell_e (f_e)}$ and $\std_{\route}(f) =\sum_{e\in\route} {\std_e(f_e)}$, for a constant $r\ge 0$  that quantifies the user \divpar, the user's cost along route $\route$ under flow $f$ is
\begin{equation}\label{eqn:pathcost}
\pathcost^r_{\route}(f) =\sum_{e\in\route} {\ell_e (f_e)} +r \sum_{e\in\route} {\std_e(f_e)}=\ell_{\route}(f)+r \std_{\route}(f)
\end{equation} 
{We} assume that for any edge and for any 
player's \divpar{ }$r$, the functions  $\ell_e$ and $\ell_e+r\std_e$ are non-decreasing.  
We note that if there is an upper bound $r_{\max}$ on the possible values of the \divpar{ }$r$, then the latter assumptions do not require  $\std_e$ to be non-decreasing.  This is desirable because,  for example, in  risk-averse selfish routing where  $\std_e$ models the variance, $\std_e$ can be a decreasing function of the flow.

\smallskip


\noindent{\bf Players Heterogeneity.} 
{We} assume that the \divpar{} distribution {$r_k$} for each commodity $k$ is heterogeneous, i.e. 
there may be more than one value of the  \divpar{s} $r$ for the players routing commodity $k$.
We use the term \emph{\ntrl} to refer to players with $r=0$.

We consider two cases: where the distribution of the \divpar{ }among the players is discrete and where it is continuous.\footnote{In fact, for our results,  we could only focus on the discrete case, though we would first have to prove that any continuous case instance has a corresponding discrete case instance such that their homogeneous equilibria have 
the same costs, as do their heterogeneous equilibria.} For a discrete distribution of, say, $n$ discrete values $r^k_1,\ldots ,r^k_n$, the demand $d_k$ is a vector $d_k=(d^k_1,\ldots d_n^k)$ where each $d^k_i$ denotes the total demand of Commodity $k$ with \divpar{ }$r^k_i$.
We let $d^k$ denote Commodity $k$'s total demand, $d^k=\sum_{i=1}^nd^k_i$.
For a continuous distribution with infimum and supremum $r^k_{\min}$ and $r^k_{\max}$ respectively, the demand $d_k$ is given by a density function 
$\rho_k:[r_{\min},r_{\max}]\rightarrow \mathbb{R}_{\geq 0}$ such that for any two values $r_1\leq r_2$ the total demand with \divpar{ }$r_1\leq r\leq r_2$ has magnitude $\int_{r_{1}}^{r_{2}}\rho_k(r)dr$, and ${d^k}=\int_{r^k_{\min}}^{r^k_{\max}}\rho_k(r)dr$. 
Variables $f^r_{\route}$ and $f_e^r$ denote the flow of \divpar{ }$r$ on path $\route$ and edge $e$, respectively.

Formally, an instance is described by the tuple $(G,\{(\ell_e,\std_e)\}_{e\in E},\{(s_k,t_k)\}_{k\in K},\{d_k\}_{k\in K},\{r_k\}_{k\in K})$ for the discrete case, where $r_k=(r^k_1,\ldots ,r^k_n)$  is the vector of different \divpar{s} encountered in the heterogeneous population, and by the tuple 
$(G,\{(\ell_e,\std_e)\}_{e\in E},\{(s_k,t_k)\}_{k\in K},\{\rho_k\}_{k\in K})$ in the continuous case.

\smallskip
\noindent{\bf Equilibrium flows.}
The  Wardrop equilibrium of an instance is a
flow $f$ such that for every $k\in K$, for every path
$\route\in\paths_k$ with positive flow, and any \divpar{ }$r$ on it, {the  path cost} 
$\pathcost_{\route}^r(f) \leq \pathcost_{\route'}^r(f)$ for all {paths} $\route' \in \paths_k\,$. 

From here on, we shall refer to the  Wardrop equilibrium as the equilibrium.
Our goal is to compare the total user cost at the equilibrium of an instance with  a population that has heterogeneous \divpar{s}, 
to
the total user cost at the {equilibrium} of the same instance but with the population of each commodity keeping its magnitude changed to be homogeneous, with \divpar{ }equal to the expected value of the \divpar{ }distribution in the heterogeneous population of the commodity. {To differentiate more easily, for} a heterogeneous instance we call the former the {\em heterogeneous equilibrium} and the latter the (corresponding) {\em homogeneous equilibrium}. 
We usually denote the heterogeneous equilibrium by $g$ and the homogeneous equilibrium by $f$. 
The existence of both equilibria is guaranteed   by e.g.  \citet[Theorem 2]{Schmeidler1973}.
We note here that in general we {do not} need uniqueness of equilibria, neither for the edge costs nor for the edge flows. 
Our results hold for any arbitrary pair of heterogeneous and homogeneous equilibria of the corresponding instances.  
Also, as for classic routing games, {without loss of generality (WLOG)} we may assume that equilibrium flows are acyclic.

\smallskip
\noindent{\bf Total Costs.}
For a heterogeneous equilibrium flow vector $g$, in the discrete case, the \emph{heterogeneous total cost} of Commodity $k$ is denoted by $C^{k,ht}(g)
= \sum_{j=1\ldots n}d^k_jc^{k,r_j^k}(g)$ where $c^{k,r_j^k}(g)$ denotes the common cost at equilibrium $g$ for players of \divpar{ }$r_j^k$ in Commodity $k$. 
In the continuous case, the \emph{heterogeneous total cost} of Commodity $k$ is denoted by $C^{k,ht}(g)
= \int_{r^k_{\min}}^{r^k_{\max}} \rho_k(r)c^{k,r}(g)dr$, where $c^{k,r}(g)$ denotes the common cost at equilibrium $g$ for players of \divpar{ }$r$ in Commodity $k$. The \emph{heterogeneous total cost} of  $g$ is then $C^{ht}(g)=\sum_{k\in K} C^{k,ht}(g)$.
For the corresponding homogeneous equilibrium flow $f$, i.e.\ the instance with \divpar{ }$\bar{r}^k$, where $\bar{r}^k$ denotes the average \divpar{ }for Commodity $k$,
players of  Commodity $k$ share the same cost $c^{\bar{r}^k}(f)$.
Then, the \emph{homogeneous total cost} of Commodity $k$ under $f$ is
$C^{k,hm}(f)=d_kc^{\bar{r}^k}(f)$,
and  the \emph{homogeneous total cost} of $f$ is
$C^{hm}(f)=\sum_{k\in K} C^{k,hm}(f)$.
Finally, if  
$C^{ht}(g)\leq C^{hm}(f)$, we say that \emph{diversity helps};
if not, we say that \emph{diversity hurts}. For our characterization to be
meaningful, we assume an \emph{average-respecting} demand, i.e.,
a demand where $\forall i, j : \bar{r}^i=\bar{r}_j$. Otherwise, diversity may hurt in simple instances, e.g., with two parallel links and two
commodities (see Appendix \ref{app:nonAvRes} for an example).

\smallskip
\noindent{\bf  Networks.}
For a network $G$ we let $V(G)$ and $E(G)$ denote its vertex set and edge set, respectively.
 
A directed $\stot$ network $G$ is \emph{series-parallel} if it consists
of a single edge $(s, t)$, or it is formed by the series or parallel composition of two series-parallel
networks with terminals $(s_1, t_1)$ and $(s_2, t_2)$, respectively.
In a \emph{series composition}, $t_1$ is identified
with $s_2$, $s_1$ becomes $s$, and $t_2$ becomes $t$.  In a \emph{parallel
composition}, $s_1$ is identified with $s_2$ and becomes $s$, and $t_1$ is
identified with $t_2$ and becomes $t$.
The internal vertices of a series-parallel network $G$ are all its vertices other than its terminals.

An $\stot$ series-parallel network may be represented using  a sequence of networks $B_j$ connected in series,  where each  $B_j$ is either a single edge or two  series-parallel networks connected in parallel\footnote{Note that this definition captures the simple case of many edges connected in parallel. }. 
Given a series-parallel network $H$, we can write $H=sB_1v_1B_2v_2\ldots B_{b-1}v_{b-1}B_{b}t$, where for any $j$ and triple $xB_jy$, $x$ and $y$ are the terminals of the series-parallel network $B_j$, 
and $B_j$ is either a single edge or a 
parallel combination of two series-parallel networks.  
We refer to the $B_j$'s as \emph{blocks}, the prescribed representation as the \emph{block representation} of $H$,
and the $v_i$'s as \emph{separators}, as they separate $s$ from $t$. Two series-parallel networks  $G_1$ and $G_2$ 
 are said to be {\em block-matching} if  for every block $B$ of $G_1$ and every block $D$ of  $G_2$, either $E(B)=E(D)$ or $E(B)\cap E(D)=\emptyset$. Note that $E(B)= E(D)$ implies that $B$ and $D$ have the same terminals and direction, as for either $B$ or $D$, the source vertex will have only outgoing edges toward the internal vertices and the target vertex will have only incoming edges from the internal vertices. 


For a  $k$-commodity network $G$, let $G_i$  be the subnetwork of $G$ that contains all the vertices and edges of $G$ that belong to a simple $s_i\dsh t_i$ path for Commodity $i$. In other words, $G_i$ is the {subnetwork} of $G$ for Commodity $i$ that equilibria flows will consider, as they are, WLOG, acyclic.  
 A multi-commodity network $G$ is \emph{block-matching} if for every $i$, $G_i$ is series-parallel, and for every $i,j$,  $G_i$ and $G_j$ are block-matching. 
 An example is given in Figure \ref{fig:blockMatched}.

\begin{figure}\center
\includegraphics[scale=0.45]{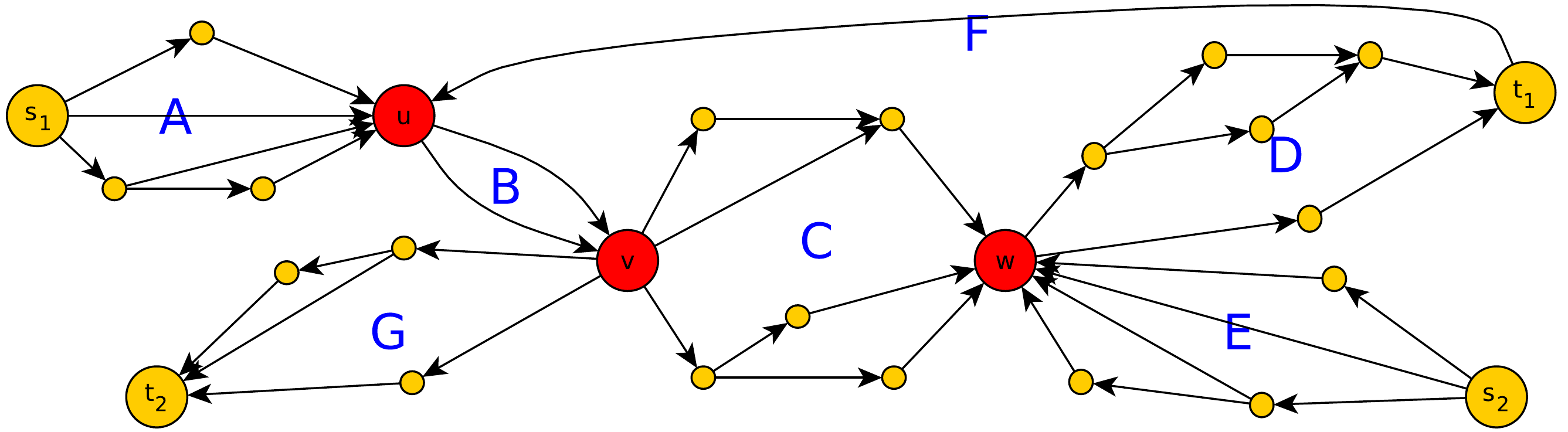} 
\caption {A block-matching network of $2$ commodities. $G_1$ and $G_2$ are series-parallel and their block representations are  $G_1=s_1AuBvCwDt_1$ and $G_2=s_2EwDt_1FuBvGt_2$. $G_1$ and $G_2$ share exactly blocks $B$ and $D$ and do not share any edge on any other of their blocks. If we add an edge from $s_1$ to $t_1$, then the network stops being block-matching since $G_1$ will be a block by itself
and it will not match any of the blocks of $G_2$.}
\label{fig:blockMatched}
\end{figure}



\section{Topology of Single-Commodity Networks for
 which Diversity Helps}\label{sec:TopSingle}

In this section, we fully characterize the topology of single-commodity networks for which, with any choice of {heterogeneous} demand and edge functions, diversity helps. 
WLOG we may restrict our attention to single-commodity networks whose edges all belong 
to some simple source-destination path as only these edges are going to be used by the (WLOG, acyclic) equilibria and thus all other edges can be discarded. It turns out that this topology is exactly that of series-parallel networks (Theorems \ref{thm:SePaSufficient} and \ref{thm:SePaNecessary}).  
{Ommited proofs can be found in Appendix~\ref{app:ommitedProofs}.}

\subsection{Series Parallel Implies Diversity is Helpful} 
Throughout this section we will be  considering
a heterogeneous instance $\G$ on an $\stot$ series-parallel network $G$. 
We let $\F$ denote the corresponding homogeneous instance.
We let $g$ denote an equilibrium flow for $\G$ and $f$ an equilibrium flow for $\F$.
Finally, we let $C^{ht}(g)$ denote the cost of flow $g$ and $C^{hm}(f)$
the cost of flow $f$.   Although redundant, we keep  the superscripts  as a further
reminder of the flow type at hand.

The key observation is that there is a path $P$ used by flow $f$ such that for every edge
in $P$, $f_e \ge g_e$, and hence 
for any $r \in [0, r_{\max}]$, 
$ c_p^r(f) \ge c_p^r(g)$ (Lemmas~\ref{lm:greaterflow}\footnote{Lemma \ref{lm:greaterflow} is similar to \cite[Lemma 2]{DBLP:journals/geb/Milchtaich06}, though for completeness we include its proof here.} and~\ref{lm:SePaPath}).
We then deduce our result: $C^{ht}(g) \le C^{hm}(f)$ 
(
Theorem~\ref{thm:SePaSufficient}).

\begin{lemma}\label{lm:greaterflow}
Let G be an $\stot$ series-parallel network and let $x$ and $y$ be flows on $G$ that route $d_1$ and $d_2$ 
units of traffic respectively, with $d_1\geq d_2$ and $d_1>0$.
Then, there exists an $\stot$ path $P$ such that for all $e \in P, x_e>0$ and $x_e\geq y_e$. 
\end{lemma}

\begin{lemma}\label{lm:SePaPath}
There exists
a path $P$ used by $f$ such that for any $r\in [0,r_{\max}]$, $c^r_p(g)\leq c^r_p(f)$.
\end{lemma}

\begin{proof}
Flows $f$ and $g$ have the same magnitude on the series-parallel network $G$. 
Applying Lemma \ref{lm:greaterflow} with $x=f$ and $y=g$ implies that there exists an $\stot$ path $P$ 
such that for all $ e \in P,$ $f_e>0$, implying that WLOG $P$ is used by $f$, and $f_e\geq g_e$. 
By assumption, for any $r\in [0,r_{\max}]$, 
$\ell_e+r\sigma_e$ is non-decreasing, 
and thus for all $e \in P,
\ell_e(f_e)+r\sigma_e(f_e)\geq \ell_e(g_e)+r\sigma_e(g_e)$. 
Consequently, $\sum_{e\in P}\big(\ell_e(f_e)+r\sigma_e(f_e)\big)\geq \sum_{e\in P}\big(\ell_e(g_e)+r\sigma_e(g_e)\big) \Leftrightarrow c^r_p(g)\leq c^r_p(f)$ as needed.

\end{proof}

\begin{theorem}\label{thm:SePaSufficient}

 $C^{ht}(g)\leq C^{hm}(f).$\footnote{The inequality might be strict. Consider the case of 2 parallel links with $(\ell_1(x),\sigma_1(x))=(1,x)$ and $(\ell_2(x),\sigma_2(x))=(2,0)$, and $1$ unit of flow, half with $r=0$ and half with $r=2$.}
\end{theorem}

\begin{proof}
Since $G$ is a series-parallel network, on setting $\bar{r}=E[r]$ and then applying Lemma \ref{lm:SePaPath}, 
we obtain that there is a path $P$ used by $f$ such that 
\begin{equation}
\label{eqn:singlePathRelation}
\ell_p (f) + \bar{r} \std_p(f)\geq \ell_p (g) + \bar{r} \std_p(g)
\end{equation} 

WLOG we can assume that the total demand $d=1$.
We {first} bound the total cost of $g$ in terms of the cost of path $P$ under $g$ 
and then we use (\ref{eqn:singlePathRelation}) 
to further bound it in terms of the cost of path $P$ under $f$. 
The latter equals the cost of $f$, as the demand is equal to $1$.

Consider the heterogeneous equilibrium flow $g$. 
By the equilibrium conditions, for any player of \divpar{ }$r$, for any $r$,
the cost she incurs with flow $g$ is
 $c^{r}(g)\leq \sum_{e \in p} {\ell_e (g_e)} +r \sum_{e \in p} {\std_e(g_e)} $.
In other words, there is no incentive to deviate to path $P$ (if not already on it). 
Thus, if the \divpar{s} are discrete, 
given by a demand vector $(d_1,\ldots,d_k)$ of \divpar{s } $(r_1,\ldots,r_k)$,
$$C^{ht}(g)
\leq 
\sum_{i=1\ldots k}d_i\Big(\sum_{e \in p} {\ell_e (g_e)} +r_i \sum_{e \in p} {\std_e(g_e)}\Big)=\ell_p(g)+\bar{r} \std_p(g),$$
with the last equality following as the total demand is 1 
and the average \divpar{ }is $\bar{r}=\sum_{i=1}^kd_ir_i$.
If instead the \divpar{s} are continuously distributed on the demand with density function $\rho(r)$, 
with $r_{\min}$ and $r_{\max}$ being their infimum and supremum respectively, 
$$C^{ht}(g)
\leq 
\int_{r_{\min}}^{r_{\max}} \rho(r)\Big(\sum_{e \in p} {\ell_e (g_e)} +r \sum_{e \in p} {\std_e(g_e)}\Big)dr=\ell_p(g)+\bar{r} \std_p(g),$$
with the last equality following as the total demand is 1, i.e. $\int_{r_{\min}}^{r_{\max}} \rho(r)dr=1$,  
and the average \divpar{ }is $\bar{r}=\int_{r_{\min}}^{r_{\max}} r\rho(r)dr$.
In both {the discrete and continuous case}, 
as $P$ is used by $f$,
we have $C^{hm}(f)= \ell_p (f) + \bar{r} \std_p(f)$, and applying 
(\ref{eqn:singlePathRelation}) we obtain

$$C^{ht}(g)\leq \ell_p(g)+\bar{r} \std_p(g)\leq \ell_p (f) + \bar{r} \std_p(f)=C^{hm}(f).$$ 
%
\end{proof}

\subsection{The Series Parallel Condition is Necessary}

To prove the necessity of the network being series-parallel, we begin by constructing an instance for which diversity hurts, 
i.e.\ the heterogeneous equilibrium has total cost strictly greater than the total cost of the homogeneous equilibrium
(Proposition~\ref{prop:BraMayHurt}).
Then, in Theorem~\ref{thm:SePaNecessary},
we show how to embed this instance into any network that is not series-parallel.

Recall the Braess graph $G_B$, shown in Figure~\ref{fig:BraGoodBad}.
\begin{proposition}\label{prop:BraMayHurt}
For any strictly heterogeneous demand on the Braess graph $G_B$,
there exist edge functions $(\ell_e)_{e\in E}$ and  $(\sigma_e)_{e\in E}$ 
that depend on the demand, for which $C^{ht}(g)>C^{hm}(f)$. In addition, this remains true if we are restricted to only using affine functions.
\end{proposition}

\begin{proof}
We may assume WLOG that the demand is of unit size.

Let $\bar{r}$ be the average \divpar{ }and let $r_{\min}$ be the infimum of the \divpar{s}' distribution.
Let $r_0$ be any \divpar{ }and let $d_0$ be  the total demand with \divpar{ }$\leq r_0$.
Suppose that in addition, $r_{\min}\leq r_0 < \bar{r}$,
and the corresponding $d_0$ satisfies $d_0>0$. 
As  the demand is strictly heterogeneous, there must be such an $r_0$. 
Later on, $r_0$  will be specified further.  

In addition, we let $h$ be any continuous, strictly increasing cost function with 
$h(\frac{1}{2})=1$ and $h(\frac{1}{2}+\frac{d_0}{2})=1+\frac{\bar{r}-r_0}{2}$.

Consider the Braess graph $G_B=(\{s,u,v,t\},\{(s,u),(u,t),(u,v),(s,v),(v,t)\})$ 
with cost functions $\ell_{(s,u)}(x)=\ell_{(v,t)}(x)=h(x)$,
$\sigma_{(s,u)}(x)=\sigma_{(v,t)}(x)=0$,  $\ell_{(u,t)}(x)=\ell_{(s,v)}(x)=2+\frac{\bar{r}+r_0}{2}$,
and $\sigma_{(u,t)}(x)=\sigma_{(s,v)}(x)=0$, and $\ell_{(u,v)}(x)=1$ and $\sigma_{(u,v)}(x)=1$.  
The instance is shown in Figure~\ref{fig:BraGoodBad}.

\begin{figure}\center
\includegraphics[scale=0.35]{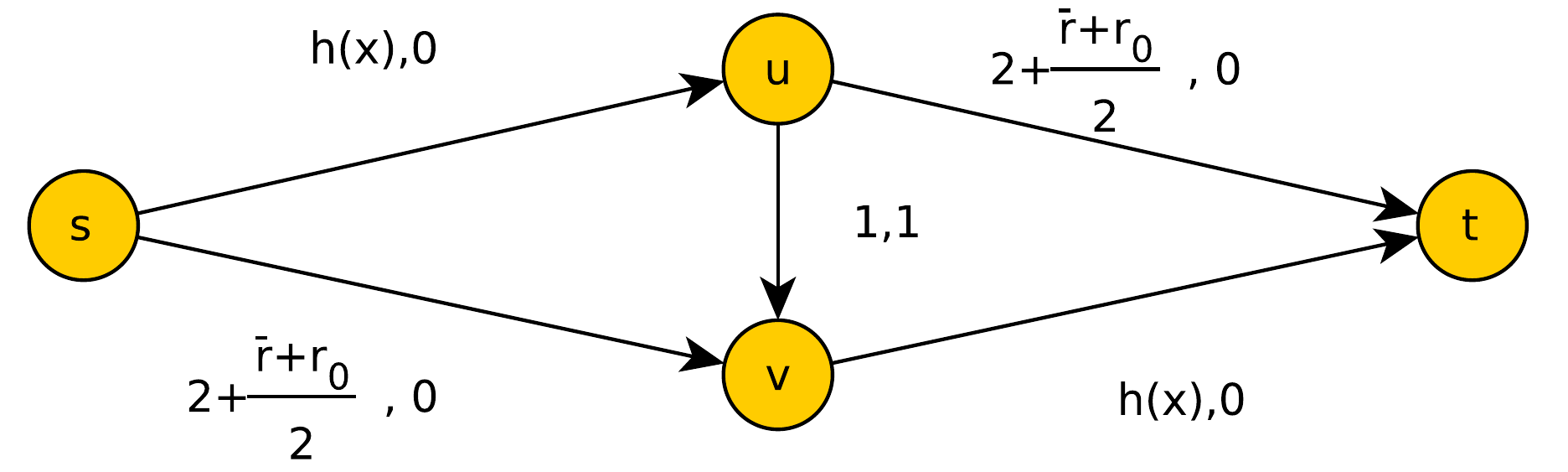} 
\caption {The Braess network with the edge functions  of Proposition \ref{prop:BraMayHurt}. The pair $a(x),b(x)$ on each edge denotes the latency and \stdev{ }functions, respectively.
}\label{fig:BraGoodBad}
\end{figure}

The heterogeneous equilibrium $g$ routes $d_0$ units of flow through the zig-zag path, 
i.e.\  path $s, u, v, t$; the rest of the flow is split between the upper and lower paths $s, u, t$ and $s, v, t$. 
This follows because with this routing, for players of \divpar{ }$r\leq r_0$, the zig-zag path costs
 $2(1+\frac{\bar{r}-r_0}{2})+1+r\leq 2(1+\frac{\bar{r}-r_0}{2})+1+r_0= 3+\bar{r}$ 
while the other paths cost $1+\frac{\bar{r}-r_0}{2}+2+\frac{\bar{r}+r_0}{2}=3+\bar{r}$, 
and for a player of \divpar{ }$r\geq r_0$, the upper and lower paths cost 
$1+\frac{\bar{r}-r_0}{2}+2+\frac{\bar{r}+r_0}{2}=3+\bar{r}$ 
while the zig-zag path costs $2(1+\frac{\bar{r}-r_0}{2})+1+r\geq 2(1+\frac{\bar{r}-r_0}{2})+1+r_0=3+\bar{r}$. 

To compute $C^{ht}(g)$, first note that players of \divpar{ }$r> r_0$, who have total demand equal to $1-d_0$,
 have cost $3+\bar{r}$, and all players of any \divpar{ }$r\leq r_0$ have cost 
$2(1+\frac{\bar{r}-r_0}{2})+1+r \geq 2(1+\frac{\bar{r}-r_0}{2})+1+r_{\min}=3+\bar{r}+r_{\min}-r_0$. 
The total cost of $g$ is thus 
$C^{ht}(g)\geq d_0(3+\bar{r}+r_{\min}-r_0)+(1-d_0)(3+\bar{r})=3+\bar{r}+d_0(r_{\min}-r_0)$.

The homogeneous equilibrium $f$ uses only the upper and lower paths. 
This follows because with this routing,
 for the average \divpar, the upper and lower paths cost 
$1+2+\frac{\bar{r}+r_0}{2}=3+\frac{\bar{r}+r_0}{2}$, 
while the zig-zag path costs $1+1+\bar{r}+1=3+\bar{r} > 3+\frac{\bar{r}+r_0}{2}$. 
The total cost of $f$ is thus $C^{hm}(f)=3+\frac{\bar{r}+r_0}{2}$.

Now we further specify $r_0$ so as to ensure $C^{ht}(g)>C^{hm}(f)$. 
By the above computations, it suffices to prove the existence of an $r_0$ that satisfies 
$r_{\min}\leq r_0<\bar{r}$ and   $d_0>0$ and in addition satisfies
  $$3+\bar{r}+d_0(r_{\min}-r_0)>3+\frac{\bar{r}+r_0}{2}\Longleftrightarrow \frac{\bar{r}-r_0}{2}+d_0(r_{\min}-r_0)>0.$$
As $r_0$ (which is $<\bar{r}$) goes to $r_{\min}$,
 the strictly positive quantity $\frac{\bar{r}-r_0}{2}$ increases 
and the non-positive quantity $d_0(r_{\min}-r_0)$ goes to $0$ 
(because $d_0$ decreases and $r_0$ goes to $r_{\min}$). 
On the other hand, by definition, $r_{\min}$ is the infimum of the \divpar{s}, 
and thus for any $\epsilon>0$, there is a positive demand with \divpar{ }$r\leq r_{\min}+\epsilon$.  
Therefore, there exists an $r_0$ satisfying the above inequality with $r_{\min}\leq r_0<\bar{r}$ and $d_0>0$, as needed.  

{The above construction can be extended to only use  affine functions. 
This can be done for example by changing function $h$ to the linear function that satisfies $h(0)=0$, 
$h(\frac{1}{2})=A$ and $h(\frac{1}{2}+\frac{d_0}{2})=A+\frac{\bar{r}-r_0}{2}$, and for that $A$
(in fact, $A= \frac{\bar r - r_0} {2 d_0}$), 
only changing $\ell_{(s,v)}$, $\ell_{(u,t)}$ and $\ell_{(u,v)}$ to $\ell_{(s,v)}(x)=\ell_{(u,t)}(x)=2A+\frac{\bar{r}+r_0}{2}$ and  $\ell_{(u,v)}(x)=A$.}
\end{proof}

\begin{theorem}\label{thm:SePaNecessary}
If $G$ is not series-parallel, then for any strictly heterogeneous  demand there are cost functions for which 
$C^{ht}(g)>C^{hm}(f)$.
\end{theorem}
We defer the proof to the appendix.
Instead, in the next section, we will enter into the more challenging construction
needed for the multi-commodity case.

\section{Topology of Multi-Commodity Networks for which Diversity Helps}\label{sec:TopMulti}

In this section we fully characterize the topology of multi-commodity networks for which, 
with any choice of {
heterogeneous} average-respecting demand and edge functions, diversity helps. 
Because of Theorem~\ref{thm:SePaNecessary}, if we require diversity to help on any instance on $G$, 
then for any commodity $i$,  $G_i$ needs to be series-parallel. 
Yet, as we shall see in Proposition~\ref{prop:JustTwoPaths}, this is not enough.
We also need to understand the overlaps of the $G_i$'s.
It turns out that the allowable overlaps are exactly captured by the topology of block-matching networks 
(Theorems~\ref{thm:MultiSePaSufficient} and~\ref{thm:blockMatchedNecessary}).
{Ommited proofs can be found in Appendix~\ref{app:ommitedProofs}.}

\subsection{Sufficiency}

Using Theorem \ref{thm:SePaSufficient}, we can obtain 
 an analogous theorem  for the multi-commodity case.

\begin{theorem}\label{thm:MultiSePaSufficient}
Let $G$ be a  $k$-commodity block-matching network. Then, for  any instance on  $G$ with average-respecting demand $C^{ht}(g)\leq C^{hm}(f).  $
\end{theorem}


 \begin{proof}
 Consider Commodity $i$ and
let $G_i=s_iB_1v_1\ldots v_{b_i-1}B_{b_i}t_i$ be its block representation.
Consider an arbitrary $B_j$ with terminals $v_{j-1}$ and $v_j$. 
Because $G$ is block-matching, any other Commodity $l$ either contains $B_j$ {as a block} in its block representation 
or contains none of its edges.
Also, recall that, as explained in the preliminaries section, if 
$G_l$ contains $B_j$, it has the same terminals $v_{j-1}$ and $v_j$. 
This implies that under any routing of the demand, either all of $l$'s demand goes through $B_j$ or none of it does. 
This means that under both equilibria $g$ and $f$, 
the total traffic routed from $v_{j-1}$ to $v_j$ through $B_j$ is the same 
 which further  implies that, if restricted to the block, the cost of the heterogeneous equilibrium is less than  or equal to that of the homogeneous equilibrium: $C^{ht}(g)\Big{|}_{B_j}\leq C^{hm}(f)\Big{|}_{B_j}$. For the latter, recall that the demand is average-respecting and thus $f$ has a single average parameter. 
On the other hand, if we let $\mathcal{B}$ be the set of all the blocks of all commodities, then $C^{ht}(g)=\sum_{B\in{\mathcal{B}}}C^{ht}(g)\Big{|}_B$ and $C^{hm}(f)=\sum_{B\in{\mathcal{B}}}C^{hm}(f)\Big{|}_B$ which using the previous inequality proves the result.
 \end{proof}

\subsection{Necessity}

To derive the necessity we first give an example of a non-block-matching network for which diversity hurts (Proposition~\ref{prop:JustTwoPaths}). Then, after proving some properties for commodities for which
 the corresponding $G_i$ are series-parallel (Lemmas~\ref{lm:SePaProperties1} and~\ref{lm:SePaProperties2}), 
we mimic the above example to obtain contradicting instances for networks that are not block-matching 
and thereby prove Theorem~\ref{thm:blockMatchedNecessary}.

\hide{
We focus on showing that if a multi-commodity network has the property that for
any edge functions 
diversity helps then it must be a block-matching network. 
For this, we first give a simple instance  on a multi-commodity network that is not block-matching  where diversity hurts. This instance will be mimicked (in some sense) 
in order to prove the desired result. 
}



Let $G$ be the following 2-commodity network, {depicted in Figure~\ref{fig:justTwoPaths}}. $G_2$, the subnetwork for Commodity 2,
consists of a simple $s_2\dsh t_2$ path $P_2$, while $G_1$, the subnetwork for Commodity 1, 
is formed from two simple $s_1\dsh t_1$ paths named $P_1$ and $P_3$;
$P_1$ and $P_2$ are disjoint, while $P_2$ and $P_3$ share a single edge, named $e_2$.
Finally $e_1$ is an edge on $P_1$ but not on $P_3$.

\begin{proposition}\label{prop:JustTwoPaths}
There exist edge functions and demands on $G$ for which diversity hurts.
\end{proposition}

\begin{proof}
Let  $d_1=d_2=1$  be the total demands for Commodities $1$ and $2$ respectively. 
Let $G_1$'s demand consist of $\tfrac 34$ \ntrl{ }players and $\tfrac 14$  players 
with \divpar{ }equal to $4$, and let $G_2$'s demand   consist of  players 
with \divpar{ }equal to $1$.
To all edges other than $e_1$ and $e_2$, assign latency and \stdev{ }functions equal to $0$. 
Assign edge $e_1$ the constant latency function $\ell_1(x)=1$ and the constant \stdev{ }function $\std_1(x)=2$.
Assign edge $e_2$ the constant \stdev{ }$\std_2=0$, and as latency function any $\ell_2$
 that is continuous and strictly increasing, with $\ell_2(1)=3$ and $\ell_2(\frac{5}{4})=9$. 

\begin{figure}\center
\includegraphics[scale=0.48]{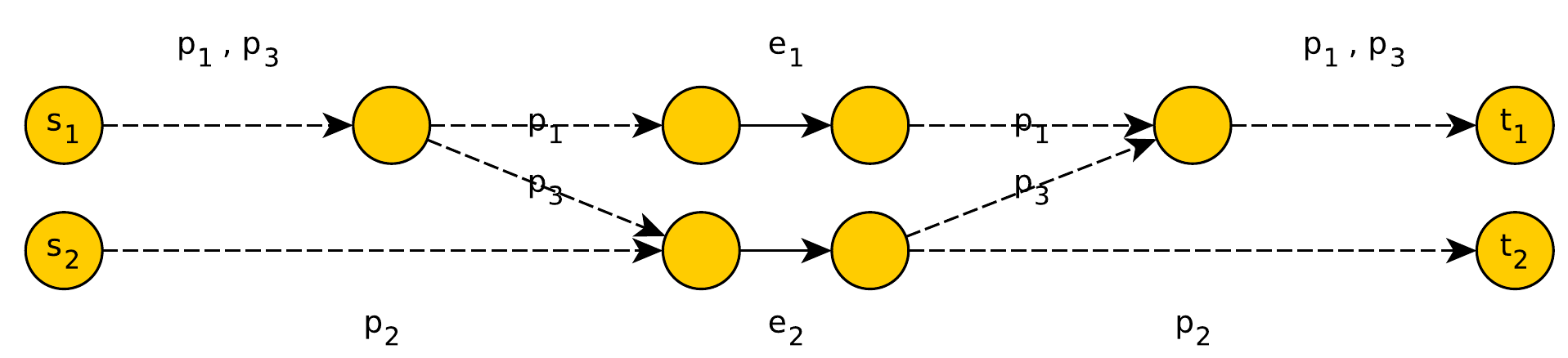} 
\caption {The network for Proposition \ref{prop:JustTwoPaths}
\label{fig:justTwoPaths}}
\end{figure}

The equilibrium costs depend only on the flow through edges $e_1$ and $e_2$, as all other edges have {cost $0$}. 
Also note that at least $1$ unit of flow will go through $e_2$ {as} this is the only route for $G_2$'s demand.

In the heterogeneous equilibrium $g$ of this instance, $\tfrac 34$ units of flow are routed through $e_1$, and  
$1+\tfrac 14$ units of flow are routed through $e_2$,
as then the \ntrl{ }players of $G_1$ compute a cost for $P_1$ equal to $1$, and a cost for $P_3$ equal to $9$,
and thus prefer $P_1$, while the remaining players of $G_1$ compute a cost equal to $9$ for both $P_1$ and $P_3$,
and thus stay on $P_3$ (recall that $\ell_2$ is strictly increasing). 
Consequently, the cost $C^{ht}(g)$ of the heterogeneous equilibrium is  
$C^{ht}(g)=1\cdot \tfrac 34 d_1 + 9 \cdot \tfrac 14  d_1+9\cdot d_2=12$.

In the homogeneous equilibrium $f$, $G_1$'s demand is all routed through  $e_1$. 
This is because the average \divpar{ }equals $1$ and thus $P_1$ and $P_3$ are both computed to cost $3$ 
(recall again that $\ell_2$ is strictly increasing). 
Thus the cost $C^{hm}(f)$ of the homogeneous equilibrium is $C^{hm}(f)=3\cdot d_1+3\cdot d_2=6$. 
Consequently,  
 $C^{ht}(g)>C^{hm}(f)$, as needed.
\end{proof}

\begin{remark}\label{rem:JustTwoPaths}
The result would still hold if the common portion of $P_1$ and $P_3$ had a positive cost instead of zero cost.
Again, it would still hold if the portion of $P_2$ after $e_2$ had a positive cost instead of zero cost.
This is close to the way we will mimic this instance in the proof of Theorem~\ref{thm:blockMatchedNecessary}. 
The idea, in both equilibria,  is to route all the flow of Commodity $1$ through two paths, 
$P_1$ and $P_3$, each containing one of $e_1$ or $e_2$, 
and to route the flow of Commodity $2$ through a path, $P_2$, that contains $e_2$. 
This is done by putting (relatively) big constants as latency functions on all the edges that depart 
from vertices of the corresponding paths up to the point where $e_1$ or $e_2$ is reached, 
though some caution is needed. 
Then, the relation of the equilibria costs will follow as in Proposition~\ref{prop:JustTwoPaths}, 
as the exact same edge functions {will} be used for edges $e_1$ and $e_2$.
This will be specified precisely when we give the construction.
\end{remark}





Next, we state 
some useful properties of series-parallel networks that are based on their block structure (the proofs are in the appendix). 
They will be used in the proof of Theorem~\ref{thm:blockMatchedNecessary}.

\begin{lemma}\label{lm:SePaProperties1}
Let $i$ be a commodity of network $G$ and suppose that $G_i$ is series-parallel.

(i)  Let $B_1$ and $B_2$ be distinct blocks of $G_i$, with $B_1$ preceding $B_2$.
There is no edge in $G$ from an internal vertex of $B_1$ to an internal vertex of $B_2$.


(ii) Let $u$ and $v$ be vertices in $G_i$. 
If $(u,v)$  is an edge of $G$ then there is a simple $s_i \dsh t_i$ path in $G_i$
that contains both $u$ and $v$ (not necessarily in that order).
\end{lemma}

\begin{lemma}\label{lm:SePaProperties2}
Let $i$ be a commodity of network $G$
and suppose that $G_i$ is series-parallel with block representation
$G_i=s_iB_1v_1\ldots v_{b-1}B_{b}t_i$.
Let $w$ be a vertex of $B_k$ for some $k\in [b]$. 

(i) Suppose that $w\neq v_{k-1}$,
and let $P$ be an arbitrary path from a vertex $u$, in a block $A$ that precedes $B_k$ in the block representation,
to vertex $w$.
Let $w'$ be the first vertex on $P$ that is an internal vertex in $B_k$, if any.
Then $P$ 
must include an edge of $B_k$ exiting $v_{k-1}$ prior to visiting $w'$.

(ii) Suppose that $w\neq v_{k}$.
Then any path of $G$ from $w$ to a vertex $u$ in a block succeeding $B_k$ has to first  
enter $v_{k}$ through one of its incoming edges that belong to $B_k$,
before going to a block $C$ that succeeds $B_k$ in the block representation.

(iii) Every simple $v_{k-1} \dsh v_k$  path in $G$ is completely contained in $B_k$.
\end{lemma}

\begin{theorem}\label{thm:blockMatchedNecessary}
Let $G$ be a multi-commodity network.
If diversity helps for every 
instance on $G$ with average-respecting
demand 
(i.e.\ for any heterogeneous equilibrium $g$ and any homogeneous equilibrium $f$,
$C^{ht}(g)\leq C^{hm}(f)$), 
then $G$ is a block-matching network.
\end{theorem}


\begin{proof}
Let $G$ have $k$ commodities. 
First, we note that for any $i\in [k]$, $G_i$ is a series-parallel network. 
{Otherwise,} by Proposition \ref{prop:BraMayHurt}, 
there is some heterogenous players' demand for Commodity $i$ and edge functions for $G_i$ such that diversity hurts. 
By letting all other commodities have zero demand we obtain an instance on $G$ for which diversity hurts, a contradiction.

{To} prove that $G$ is block-matching, it remains to show that for any two commodities $i$ and $j$ of $G$, 
for any block $B$ of $G_i$ and any block $D$ of $G_j$, either $E(B)=E(D)$ or $E(B)\cap E(D)=\emptyset$. 
To reach a contradiction we assume otherwise, i.e.\ WLOG we assume that for Commodities $1$ and $2$ 
there exist two blocks $B$ of $G_1$ and $D$ of $G_2$ that share some common edge, and at the same time,
WLOG, there is an edge in $B$ {that} is not in $D$. 
The latter implies that $B$ is not a single edge, and thus it must be a parallel combination of two series-parallel networks.


Let $u$ and $v$ be the endpoints of $B$. 
We first prove that all simple $s_2 \dsh t_2$ paths of $G_2$ that share an edge with $B$
first traverse an edge starting at $u$ before traversing any other edge of $B$ (Proposition~\ref{prop:hitEdgeu}). 
Then we prove that all $s_2 \dsh t_2$ simple paths of $G_2$, that share an edge with $B$, 
reach $u$ before traversing any internal vertex of $B$ (Proposition~\ref{prop:hitVertexu}). 
Since $E(B)\cap E(D)\neq \emptyset$, there is a simple $s_2 \dsh t_2$  path of $G_2$ that shares an edge with $B$. Proposition~\ref{prop:hitVertexu} implies that this path, $Q$, 
has a subpath consisting of a simple $s_2 \dsh u$ path $Q_1$  that shares no internal vertex with $B$. 
A completely symmetric argument shows that $Q$ has a subpath consisting of a simple $v \dsh t_2$ path $Q_3$ 
that shares no internal vertex with $B$.\footnote{For the symmetric argument, 
simply reverse all the arcs and the directions of the demand.} 
But then, for any simple $u \dsh v$ path $Q_2$ inside $B$, 
the path $Q'=Q_1, Q_2, Q_3$ is a simple $s_2\dsh t_2$ path, and thus it belongs to $G_2$. 
But this implies that all the edges of $B$ belong to $G_2$ and
 %
because $B$ is a block, these edges will all be in a single block of $G_2$. 
This block must be block $D$, since by assumption $E(B)\cap E(D)\neq \emptyset$,
contradicting the existence of an edge in $B$ and not in $D$.  
Therefore, once these propositions are proved, the theorem will {follow}.

The proofs of these propositions rely on the same idea. 
For each proposition, assuming that it does not hold, we construct instances, i.e.\ 
we choose demand and edge functions for $G$, such that diversity hurts, 
contradicting the assumption that for any instance on $G$ diversity helps. 
The construction of the contradicting instances 
is based on Remark~\ref{rem:JustTwoPaths}, which 
follows Proposition~\ref{prop:JustTwoPaths}.

\begin{proposition}\label{prop:hitEdgeu}
Let $P$ be a simple $s_2 \dsh t_2$ 
path in $G_2$ which 
shares an edge with $B$.
The first edge on $P$ in $B$ departs from $u$, i.e.\ has the form $(u,x)$ for some $x$ in $B$.
\end{proposition}

\begin{proof}
Let $B$ be the parallel combination of $H_1$ and $H_2$.
WLOG we may assume that $P$ only visits vertices of $G_1$, plus $s_2$ and $t_2$, 
as we may treat subpaths of $P$ that have vertices that lie outside $G_1$ as simple edges.
Let $w$ be the first internal vertex of $P$ that belongs to $B$, and WLOG suppose that $w$ lies in $H_1$.  
By Lemma~\ref{lm:SePaProperties1}(ii), the edge of $P$ exiting $w$ will either go toward $t_1$, i.e.\ forward, 
and thus traverse an edge of $B$ for the first time (recall also Lemma~\ref{lm:SePaProperties2}(ii)),
or will go toward $s_1$, i.e.\ backward, 
either staying in $H_1$ or going back to one of the preceding blocks of $B$.
If it goes to one of the preceding blocks of $B$, then by Lemma~\ref{lm:SePaProperties2}(i),
 it has to traverse an edge of $B$ departing from $u$ in order to re-enter the internal portion of $B$ 
(recall that $P$ has some edge in $B$) and then the proposition would hold. 
The remaining possibility is that the backward edge leads to another internal vertex of $H_1$.
However, we can only repeat this process finitely often
so if the proposition does not hold,
it must be that $P$ eventually 
traverses a first edge in $B$ that departs from an internal vertex of $H_1$. 
In this case
we will reach a contradiction by creating an instance where diversity hurts. 
This instance will be based on the instance of Proposition~\ref{prop:JustTwoPaths}. 

\begin{figure}[t]\center
\includegraphics[scale=0.44]{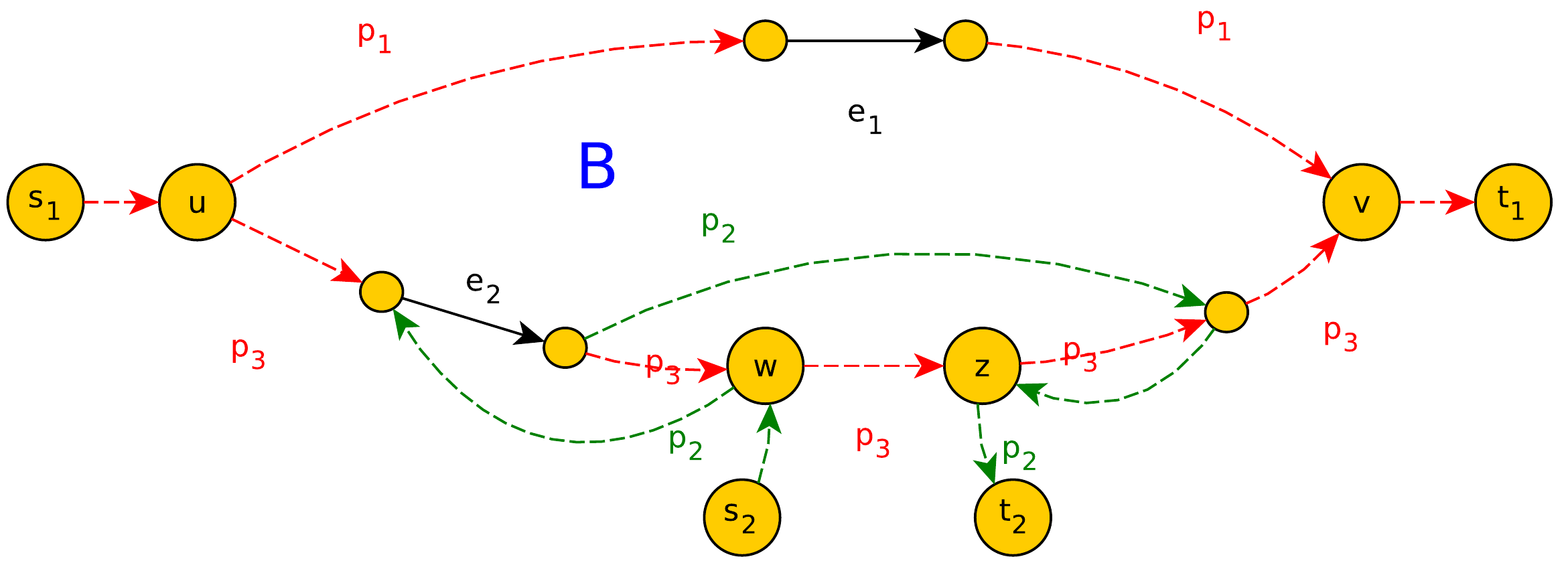}
\caption{Illustrating why $P_2 \ne P$ in general in Proposition~\ref{prop:hitEdgeu}}
\label{fig:Newexample}
\end{figure}

We would like to use the following construction at this point.
Let $P_2$ be the path $P$ resulting from the discussion in the previous paragraph
and let $e_2$ be the first edge on $P$ that lies in $B$.
Then let $P_3$ be an $s_1\dsh t_1$ path through $e_2$.
Recall that $e_2$ lies in $H_1$.
Now let $P_1$ be an $s_1\dsh t_1$ path that goes through $H_2$ and let $e_1$
be an arbitrary edge on $P_1$ in $B$.
The intention is to force the $s_1\dsh t_1$ flow to use just paths $P_1$ and $P_3$,
while the $s_2\dsh t_2$ flow uses just path $P_2$, {at the same time} ensuring that diversity is harmful
as in Proposition~\ref{prop:JustTwoPaths}.
{Consider the following edge functions}.
$e_1$ and $e_2$ receive the same edge functions as in Proposition~\ref{prop:JustTwoPaths}.
The other edges all receive a 0 \stdev.
For their latency functions, edges on $P_1$ and $P_3$ that are in $B$ receive 0 functions.
Outedges from $P_1$ and $P_3$ that lie in $B$ all receive functions of constant value
$N$ even if they are on $P_2$.
All as yet unassigned edges on $P_2$ receive 0 functions, and the remaining edges are all given functions
of constant value $M \gg N$.
However, the example in Figure~\ref{fig:Newexample} shows {that} there is a zero cost $s_2\dsh t_2$
path ($s_2,w,z,t_2$), which defeats the construction.

We fix this problem by defining the path $P_2$ as follows.
Let $x$ be the first vertex on path $P$ (in the example, this is $w$) such that there is an edge $(x,y)$ in $B$ and such that there
is a $y\dsh t_2$ path $P_{y\dsh t_2}$ which does not go through any ealier vertex on $P$ (i.e.\ any vertex from $s_2$ to $x$
inclusive). Then $(x,y)$ is chosen to be $e_2$, and $P_2$ is defined to be the simple path comprising
the initial portion of $P$ up to $x$, followed by $e_2$, followed by $P_{y\dsh t_2}$ (it may be that $P_2 = P$).
Now the above cost functions, modulo a few details, will achieve the desired contradiction.
These details follow. 


Let path $P_2$ be the simple $s_2\dsh t_2$ path that follows $P$ up to $w$ 
and keeps following it after $w$ until for the first time it finds an edge of $B$ (and hence of $H_1$) 
that departs from a vertex of $P$ and can lead to $t_2$ via a path $P'$
without returning to any of the previously visited vertices. 
Choose $e_2$ to be this edge and have $P_2$ follow $e_2$ and then path $P'$ to $t_2$. 
Note that $P$ itself can be $P_2$. 
Consider an arbitrary $s_1\dsh t_1$ simple path $P_3$ that goes through $e_2$ (and thus through $H_1$) 
and an arbitrary $s_1 \dsh t_1$ simple path $P_1$ that goes through $H_2$,
and let $e_1$ be any edge of $P_1$ in $H_2$. 
These paths and edges are going to mimic the ones of Proposition~\ref{prop:JustTwoPaths}. 
By using some large values $N$ and $M$ as latency functions for some of the edges, at both equilibria, 
we will ensure that all the flow of $G_1$ goes through the subpaths of $P_1$ and $P_3$ in $B$,
and all the flow of $G_2$ goes through the subpath of $P_2$ that starts at $s_2$ and ends with $e_2$.

\begin{figure}[t]\center
\includegraphics[scale=0.458]{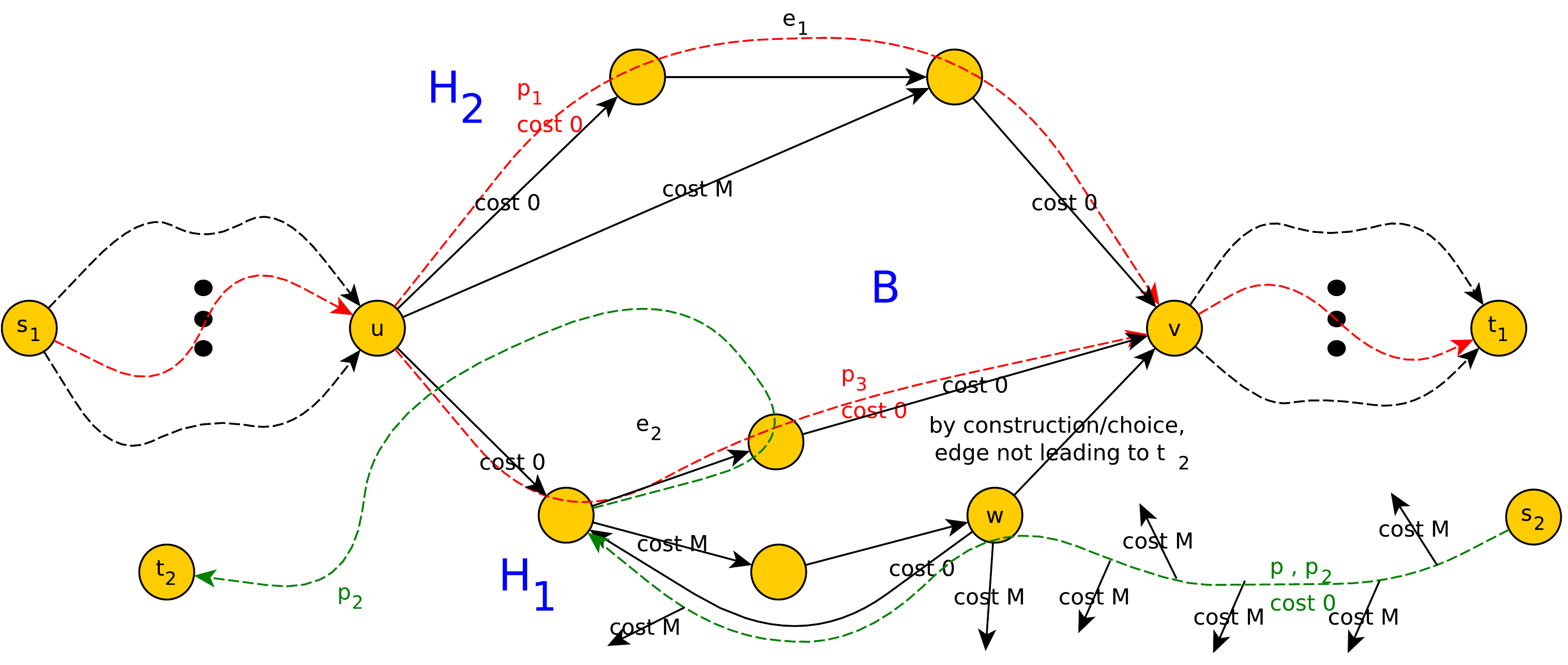}
\caption {Sample network for Proposition \ref{prop:hitEdgeu}\label{fig:prop1}}
\end{figure}

Let  $d_1=d_2=1$ be the total demands for Commodities $1$ and $2$ respectively,
 and let all other commodities have $0$ demand. 
Let $G_1$'s demand consist of $\tfrac 34$ \ntrl{ }players and $\tfrac 14$ players 
with \divpar{ }equal to $4$, and let the demand of $G_2$ consist of players 
with \divpar{ }equal to $1$.  
Assign edge $e_1$ the constant latency function $\ell_1(x)=1$
and the constant \stdev{ }function $\std_1(x)=2$. 
Assign edge $e_2$ the constant \stdev{ }$\std_2=0$, and as latency function any $\ell_2$ 
that is continuous and strictly increasing with $\ell_2(1)=3$ and $\ell_2(\frac{5}{4})=9$. 
Assign all other edges the constant \stdev{ }$\std_2=0$.
Assign to all other edges of $P_1$ and $P_3$ that lie inside $B$ 
$0$ latency functions. 
To all edges that depart from a vertex of $P_1$ or $P_3$  and that lie on $P_2$, 
assign latency functions equal to some big constant $N$, say $N=24$
(i.e.\ double the heterogeneous equilibrium cost of Proposition~\ref{prop:JustTwoPaths}). 
For all remaining edges on $P_2$, assign $0$ latency functions. 
Finally, to all remaining edges, assign constant latency functions equal to $M$,
where $M$ is defined to be $2|V(G)|\cdot N$.

\hide{
Starting from $s_2$, for all edges of $P_2$ up to but not including  $e_2$,  
assign latency and standard deviation functions equal to $0$. 
Starting from $s_2$ and  following $P_2$ up to the starting vertex of $e_2$, 
for all edges that depart from a vertex of $P_2$ but have not yet been assigned edge functions, 
In addition, to the edges departing from the starting vertex of $e_2$ 
overwrite and assign latency and standard deviation functions equal to $M$. 
To all other edges assign $0$ latency and standard deviation functions. 
Note that by the choice of the paths the above assignment assigns only one latency and standard deviation 
to each edge as $P_2$ shares no edge with $B$ before it reaches $e_2$ .
}

Note that all edges other than $e_2$ have constant edge functions. 
Thus for Commodity $1$ under both equilibria there will be a common cost $C^{B^-}$ 
that will be paid on blocks other than $B$. 
Futhermore, for Commodity $2$, any path that costs less than $M$ will follow path $P_2$ up to and including $e_2$,
for any edge leaving path $P_2$ either has cost $M$ or is an edge of $P_3$ and by construction
these edges do not create a cycle-free path to $t_2$.
Thus if any of the homogeneous or heterogeneous equilibria is to cost less than $M$
 then all the flow of Commodity $2$ will go through $e_2$,
 and from there follow a least cost path to $t_2$ 
(recall all edges other than $e_2$ have constant edge functions)
with cost $C^{e_2^-}$ say.
Note that $C^{e_2^-}$ will not be more than the cost of path $P_2$ following edge $e_2$,
which is bounded by
$|V(G)|N=\frac{M}{2}$,
and thus the path portion with cost  $C^{e_2^-}$ is preferable to any path with an edge of cost $M$.
\footnote{\label{footnt:EdgesUpperCost} For edge $e_2$ this will hold at both equilibria
but in any case we can define its latency functions so that this holds in general, e.g.\ also set $\ell_2(2)=N$.}

For the heterogeneous equilibrium $g$ of this instance, for Commodity $1$ 
route all the flow through the shortest $s_1\dsh u$ and $v\dsh t_1$ paths,
and inside $B$ route $\tfrac 34$ units of flow through $P_1$ and $\tfrac 14$ units of flow through $P_2$,
and route all the flow of commodity $2$ through $e_2$,
via $P_2$ up to $e_2$, and after $e_2$ 
via a least cost path to $t_2$. 
The \ntrl{ }players of $G_1$ compute a cost for $P_1$ equal to $1+C^{B^-}$ 
and a cost for $P_3$ equal to $9+C^{B^-}$,
 and thus prefer $P_1$ to $P_3$,
 while the remaining players of $G_1$ compute a cost equal to $9+C^{B^-}$ for both $P_1$ and $P_3$,
 and thus prefer $P_3$ to $P_1$ (recall that $\ell_2$ is strictly increasing). 
The other paths have cost at least $N+C^{B^-}=24+C^{B^-}$ and thus are not preferred by any type  of players. 
The players of Commodity $2$ pay cost equal to $9+C^{e_2^-}$ (the cost of $e_2$ plus the cost after it) 
and thus prefer staying on $e_2$ rather than paying at least $M$ to avoid it (recall $C^{e_2^-}\leq \frac{M}{2}$). 
Note that on $P_2$ for the vertices before $e_2$ there might be edges leaving $P_2$ that cost $0$ 
(these are edges of $P_3$) 
but by the definition of $P_2$ they cannot lead to $t_2$ without visiting preceding vertices. 
Putting it all together, this routing is indeed the heterogeneous equilibrium with cost  
$C^{ht}(g)=1\cdot \tfrac 34 d_1+ 9\cdot \tfrac 14 d_1+d_1C^{B^-}+9d_2+d_2C^{e_2^-}=12+C^{B^-}+C^{e_2^-}$.

For the homogeneous equilibrium $f$ of this instance, route all the demand of $G_1$ through  $e_1$, 
via $P_1$ and the shortest $s_1\dsh u$ and $v \dsh t_1$ paths, 
and route all the flow of Commodity $2$ through $e_2$, 
via $P_2$ up to $e_2$ and from there via a least cost path to $t_2$. 
The average \divpar{ }for Commodity $1$'s demand equals $1$,
 and thus $P_1$ and $P_3$ are both computed to cost $3+C^{B^-}$ while all other paths cost at least $N+C^{B^-}=24+C^{B^-}$ and thus are avoided. 
In the same way as above, the Commodity $2$ players pay cost equal to $3+C^{e_2^-}$ 
(the cost of $e_2$ plus the cost of the path portion following $e_2$) 
and thus prefer staying on $e_2$ rather than  paying at least $M$ to avoid it
(recall $C^{e_2^-}\leq \frac{M}{2}$). 
Thus the cost $C^{hm}(f)$ of the homogeneous equilibrium is 
$C^{hm}(f)=3\cdot d_1+d_1\cdot  C^{B^-}+3\cdot d_2+d_2\cdot C^{e_2^-}=6+C^{B^-}+C^{e_2^-}$. 
Consequently, $C^{ht}(g)>C^{hm}(f)$, as needed.
 
The above instance contradicts the assumption that $G$ has the property
that diversity helps for all 
edge functions. 
Thus the hypothesis that $P$, after reaching $w$ 
(and possibly moving backward while staying in the internal portion of $H_1$) 
first traverses an edge of $H_1$ (departing from an internal vertex) does not hold. 
This was what needed for the proposition to hold. 
\end{proof}

\begin{proposition}\label{prop:hitVertexu}
All simple $s_2\dsh t_2$ 
paths of $G_2$ that share an edge with $B$ reach $u$ before 
any internal vertex of $B$.
\end{proposition}
\begin{proof}
Consider an arbitrary simple $s_2\dsh t_2$ path $P$ that shares some edge with $B$. 
WLOG we may assume, again, that $P$ only visits vertices of $G_1$ and $s_2$, $t_2$, 
as we may handle subpaths of $P$ that have vertices that lie outside $G_1$ as edges. 
Let $B$ be the parallel combination of $H_1$ and $H_2$ and, assuming that the proposition does not hold, 
let $w$ be the first vertex of $P$ (before it reaches $u$) that belongs to the internal portion of $B$, 
and suppose WLOG that $w$ lies in $H_1$. 
By Lemma~\ref{lm:SePaProperties1}(ii) and Proposition \ref{prop:hitEdgeu}, the edge of $P$ exiting $w$ cannot go toward $t_1$, i.e.\ forward, 
as then it would traverse an edge of $B$ for the first time  that does not depart from $u$ (recall also Lemma~\ref{lm:SePaProperties2}(ii)).
Thus, it has to go toward $s_1$, i.e.\ backward, 
either staying in $H_1$ or going back to one of the preceding blocks of $B$. If by going backward it stays in $H_1$ then by the same argument it has again to move backward.
However, we can only repeat this process finitely often and eventually after possibly hitting some vertices of $H_1$ other than $w$ and  after possibly visiting vertices of blocks that precede $B$, $P$ hits $u$. Note that this happens  without $P$ having hit any vertex of $H_2$ prior to hitting $u$ (recall Lemma \ref{lm:SePaProperties2}(i)). 
 %
%
There are two cases.

The first case occurs when on traversing $P$ up to $u$,  
there is some edge of $B$ departing from $u$ that does not lead to $t_2$ 
without revisiting one of the previously visited vertices. 
Yet, since $P$ shares an edge with $B$,
by Proposition~\ref{prop:hitEdgeu}
there is an edge of $B$ that departs from $u$ and leads to $t_2$ 
without revisiting one of the previously visited vertices. 
Let $e_1$ be the first of the above edges and $e_2$ be the second one. 
For these edges, the same contradicting instance as in Proposition~\ref{prop:hitEdgeu} can be constructed. 
Path $P_1$ is an arbitrary simple $s_1\dsh t_1$ path containing $e_1$, 
path $P_3$ is an arbitrary simple $s_1 \dsh t_1$ path containing $e_2$,
and path $P_2$ is constructed by following $P$ up to $u$ 
(instead of some internal vertex of $H_1$ from which $e_2$ was departing) 
and from there taking $e_2$ and then a path that leads to $t_2$ without revisiting vertices. 
The edge function assignment will be exactly the same.  
Diversity hurting, and thus the contradiction, will follow in the same way as in Proposition~\ref{prop:hitEdgeu}.
 
The second and more interesting case occurs when 
on following $P$ up to $u$, all edges of $B$ departing from $u$ can lead to $t_2$ 
without revisiting one of the previously visited vertices. 
Let $e_2$ be an edge of $H_2$ (departing from $u$) with this property, 
let $P_2$ be the simple path that follows $P$ up to $u$ 
and then follows some path through $e_2$ to go to $t_2$,
and let $P_3$ be a simple $s_1\dsh t_1$ path that follows an arbitrary $s_1\dsh u$ path and 
an arbitrary $v\dsh s_2$ path,
and between $u$ and $v$ follows a path that contains $e_2$. 
Let $P_1$ be any path that follows an arbitrary $s_1\dsh u$ path and an arbitrary  $v\dsh t_2$ path,
and between $u$ and $v$ follows a path that contains $w$, and therefore goes through $H_1$. 
Let $e_1$ be the edge of $P_1$ that departs from $w$. 
Note that, because of Proposition~\ref{prop:hitEdgeu},  
$e_1$ cannot lead to $t_2$ with a simple path, i.e.\ without visiting vertices on $P$ before $w$. 
 This is a key fact for the contradiction to come and relates to the extra caution needed for this proof in comparison with the proof of Proposition \ref{prop:hitEdgeu}. See Figure~\ref{fig:prop2App} for a high level description of the instance. The rest of the details can be found in the appendix, Section  \ref{app:propHitVertexProof}.
 %
%
\end{proof}
\end{proof}

\section*{Acknowledgments}
We thank the anonymous referees for their thoughtful feedback that helped  improve this work.

\bibliographystyle{named}
\bibliography{DiverseBiblio6}

\appendix

\section*{APPENDIX}

\setcounter{section}{0}

\section{Ommited proofs}\label{app:ommitedProofs}

\subsection{An instance without average-respecting demand where diversity hurts}\label{app:nonAvRes}

\begin{figure}[h]\center
\includegraphics[scale=0.25]{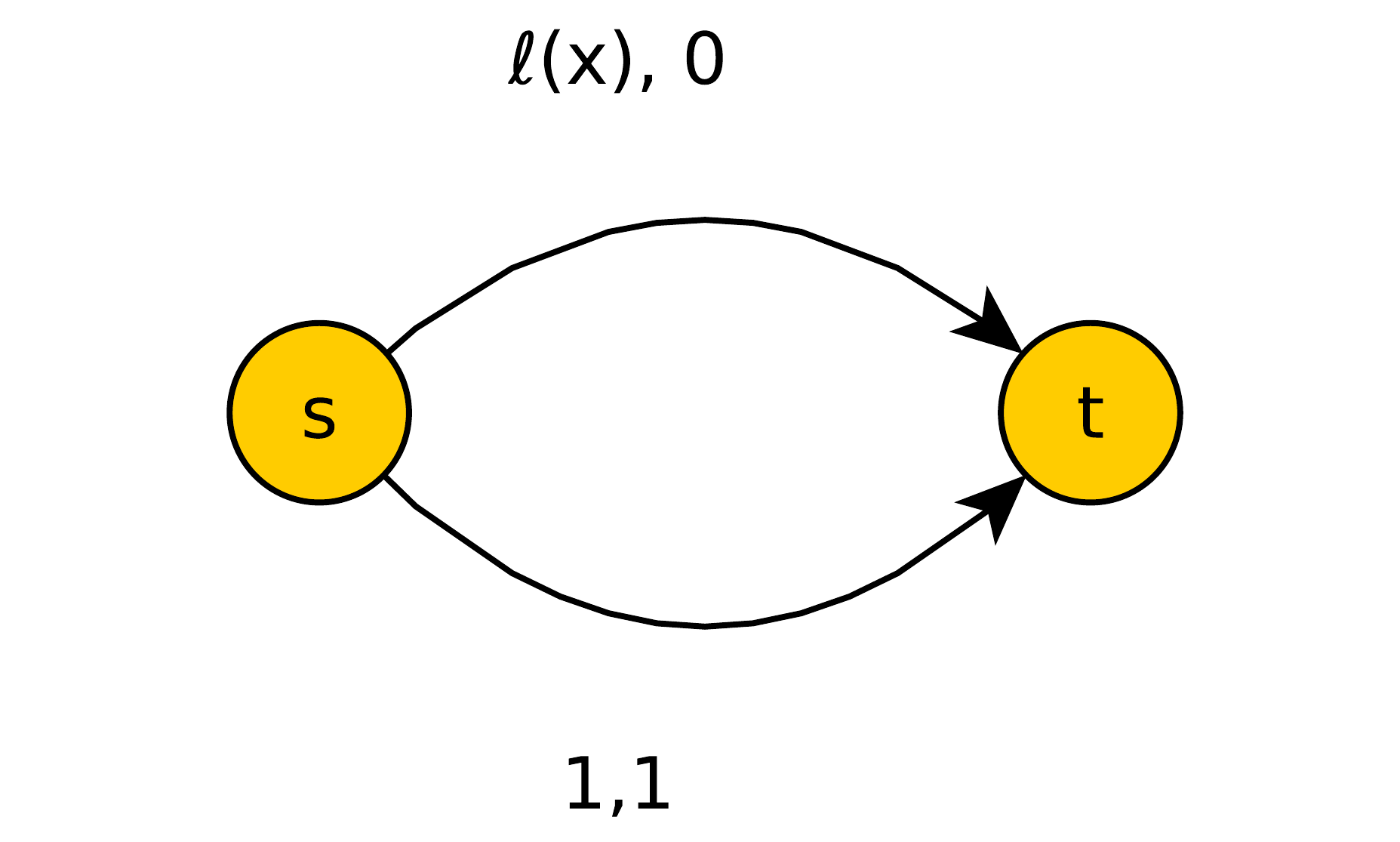}

\label{fig:Newexample}
\end{figure}
Let $\ell$ be any strictly increasing function with $\ell(1)=2$ and $\ell(10/9)=10$. Let there be two commodities, each of unit demand. The first commodity has demand that consists of homogeneous players with parameter equal to $10$. The second commodity has demand that consists of 1/9 players with parameter equal to $9$ and $8/9$ players with parameter equal to $0$. The heterogeneous equilibrium has cost equal to $12$ (only \ntrl{ }players are routed through the lower edge) while the homogeneous equilibrium has cost equal to $4$ (the first commodity uses only the upper edge and the second uses only the lower edge).

\hide{
\subsection{The proof of Lemma \ref{lm:SqrtProp}}\label{app:lm:SqrtProp}

\begin{gather}
\nonumber
x-y+\sqrt{M+y^2-x^2}\geq \sqrt{M} \\
\nonumber
\Longleftrightarrow \\
\nonumber
\cancel{x^2}-2xy+y^2+2(x-y)\sqrt{M+y^2-x^2}+\cancel{M}+y^2-\cancel{x^2}\geq \cancel{M} \\
\nonumber
\Longleftrightarrow\\
\nonumber
-2y(x-y)+2(x-y)\sqrt{M+y^2-x^2}\geq 0 \\ 
\nonumber
\Longleftrightarrow\\
\nonumber
\sqrt{M+y^2-x^2}\geq y
\\
\nonumber
\Longleftrightarrow
\\ 
\label{eqn:final-sqrt-relat}
 M-x^2\geq 0.
\end{gather}
\eqref{eqn:final-sqrt-relat} holds by hypothesis. 
%

}

\subsection{The proof of Lemma \ref{lm:greaterflow}}
\label{app:lm:greaterflow}

\hide{
We will consider only the edges of $G$ that are used under $x$ or $y$ and we will prove that in that subnetwork there exists an $\stot$  path $p$  such that $\forall e \in p: x_e\geq y_e$. All edges being used by $x$ or $y$ implies that $\forall e \in p: x_e>0$ and surely any $\stot$ path of a subnetwork of $G$ is an \stot path of $G$. 

That being said, WLOG we assume that all edges of $G$ are used by $x$ or $y$. 
}
The proof is by induction on the decomposition of the series-parallel network. 
The base case of $G$ being a single edge $e$ is trivial as $x_e=d_1\geq d_2=y_e$.

For the inductive step, first suppose that $G$ is a series combination of two series-parallel networks $G_1$ and $G_2$. 
For $i=1,2$, let  $x^i$ be the restriction of flow $x$ to $G^i$,
and $y^i$ the restriction of flow $y$. 
By the inductive hypothesis, 
there is an $\stot$ path $P_i$ in $G_i$ such that for all $e \in P_i, x^i_e\geq y^i_e$. 
It suffices to set $P$ to be the concatenation of $P_1$ with $P_2$.
 
Now assume that $G$ is a parallel combination of two series-parallel networks $G_1$ and $G_2$. 
Again, for $i=1,2$, let $x^i$ be the restriction of flow $x$ to $G^i$, and $y^i$ of flow $y$.
We may assume WLOG that the flow $d_1^1$ that $G_1$ receives in $x^1$ is at least as large
as the flow $d_2^1$ that it receives in $y^1$, and further that $d_1^1 >0$.
By the inductive hypothesis applied to $G_1$ with demands $d_1^1$ and $d_2^1$,
we obtain that there exists an $\stot$ path $P$ such that for all
$ e \in P, x^1_e\geq y^1_e$
and this implies that for all $ e \in P, x_e\geq y_e$, as needed.

\subsection{The proof of Theorem \ref{thm:SePaNecessary}}
\label{app:thm:SePaNecessary}

If $G$ is not series-parallel then the Braess graph can be embedded in it 
(see e.g. \cite{DBLP:journals/geb/Milchtaich06} or \cite{Valdes:1979:RSP:800135.804393}). 
Thus, starting from the Braess network $G_B$, by subdividing edges, 
adding edges and extending one of the terminals by one edge, we can obtain $G$. 
Fix such a sequence of operations. 
For the given heterogeneous demand, we start from the Braess instance given by Proposition \ref{prop:BraMayHurt} 
and apply the sequence of operations one by one.
Each time an edge addition occurs, 
we give the new edge a constant latency function equal to some large $M$ and \stdev{ }equal to $0$, 
each time an extension of the terminal occurs we give the new edge a constant latency and \stdev{ }equal to $0$, 
while each time an edge division occurs, if it is an edge with latency function $M$ 
we give both edges latency function equal to $M$ and \stdev{ }equal to $0$, 
otherwise we give one of the two edges the latency and the \stdev{ }functions of the edge that got divided 
and we give the other one a constant latency and \stdev{ }equal to $0$. 

It is not hard to see that taking $M=2(3+\bar{r})$, i.e.\ 
more than double the heterogeneous cost of the instance of Proposition \ref{prop:BraMayHurt},
(or $M=2(3A+\bar{r})$ if we must only use affine  functions),
 suffices to ensure that all edges having latency $M$ and \stdev{ }$0$
receive zero flow in both the heterogeneous equilibrium $g$ and the homogeneous equilibrium $f$. 
In more detail, the only $\stot$ routes that may have cost $<M$ 
are those that starting from $s$ reach, with zero cost, 
some $s'$ that corresponds to the $s$ of the instance of Proposition~\ref{prop:BraMayHurt}, 
follow some path that corresponds to one of the upper, zig-zag, or lower paths of the instance of 
Proposition~\ref{prop:BraMayHurt}, with corresponding cost, 
reach some $t'$ that corresponds to $t$ of the instance of Proposition~\ref{prop:BraMayHurt}
and from there reach $t$ with zero cost. 
This can be formally proved by induction on the number of embedding steps.
Thus,  $C^{ht}(g)>C^{hm}(f)$ can be derived in the same way as in Proposition~\ref{prop:BraMayHurt}
and that is enough to prove the theorem.

\subsection{The proof of Lemma \ref{lm:SePaProperties1}}
\label{app:lm:SePaProperties1}

For (i), if there were such an edge then a simple $s_i\dsh t_i$ path would be created that avoids the separators that lie 
between $B_1$ and $B_2$, contradicting the definition of $G_i$'s  block structure. 


For (ii),  let $u$ and $v$ be two vertices in $G_i$ such that there is no simple $s_i \dsh t_i$ path in $G_i$ 
that contains both of them. 
This implies that there is no edge between them in $G_i$.
Now, in the series-parallel decomposition of $G_i$, let $B$ be the smallest series-parallel subnetwork 
containing both $u$ and $v$. 
By the choice of being smallest and the fact that there is no edge in $G_i$ between $u$ and $v$, 
$B$ must be a composition of a $B_1$ containing $u$ and a $B_2$ containing $v$. 
$B_1$ and $B_2$ are not connected in series because then there would be an $s_i\dsh t_i$ path in $G_i$ 
containing both $u$ and $v$. 
Therefore, $B_1$ and $B_2$ are connected in parallel;
thus there cannot be any edge in $G$ between $u$ and $v$
or else it would belong in some simple $s_i \dsh t_i$ path 
and therefore belong to $G_i$, violating $G_i$'s series-parallel structure.

\subsection{The proof of Lemma \ref{lm:SePaProperties2}}

\label{app:lm:SePaProperties2}

The proofs of (i) and (ii) are by induction on the length of path $u \dsh w$. 

For (i), if the path has length equal to $1$ then it is a simple edge, i.e.\ edge $(u,w)$, and thus $u=v_{k-1}$, 
because of Lemma~\ref{lm:SePaProperties1}(i), and then (i) holds. 

Now suppose inductively that the result holds for paths of length up to $l-1$.
Let $u \dsh w$ be a path of length $l$. 
Let $(u,x)$ be the first edge on this path.
Note that by Lemma~\ref{lm:SePaProperties1}(i),
$x$ cannot belong to any successor of $B_k$ 
(unless $u=v_{k-1}$ and $x=v_k$, in which case (i) would hold). 
If $x$ belongs to $B_k$ and is not $v_{k-1}$, then by Lemma \ref{lm:SePaProperties1}(i),
$u=v_{k-1}$ and thus (i) holds. 
If $x=v_{k-1}$ or $x$ does not belong to $B_k$
(which implies it belongs to a predecessor of $B_k$)
then the inductive hypothesis holds for the length $l-1$ path $x \dsh w$,
yielding the desired edge exiting $v_{k-1}$. 
Thus (i) also holds for path $u \dsh w$.

For (ii), if the path has length equal to $1$ then it is a single edge, i.e.\ edge $(w,u)$, 
and thus by Lemma \ref{lm:SePaProperties1}(i)$u=v_{k}$, and (ii) holds. 
%
Now suppose inductively that the result holds for paths of length up to $l-1$.
Let  $w\dsh u$ be a path of length $l$. 
Let $(w,x)$ be the first edge on this path.  
Either $x=v_k$ or $x$ is an internal vertex in $B_k$, by Lemma \ref{lm:SePaProperties1}(i). 
If $x=v_k$ then (ii) holds. 
If $x$ is an internal vertex in $B_k$, then the inductive hypothesis applies t the length $l-1$ path $x \dsh w$.
Thus t(ii) also holds for path $u\dsh w$.

(iii) follows from (i) and (ii). 
Consider an arbitrary simple $v_{k-1}\dsh v_k$ path $P$. 
Path $P$ does not contain a vertex from any preceding block, for if it did, then to re-enter $B_k$ so as to reach $w$, 
according to (i), it would go through $v_{k-1}$ again, and then it would not be a simple path.
Also, aside its endpoints, $P$ does not contain a vertex from any succeeding block, for if it did, then to leave $B_k$, 
according to (ii), it would reach $v_k$ before reaching it again at the end, and then it would not be a simple path.
Thus, a path that follows some simple $s_i \dsh v_{k-1}$ path,
then follows $P$ and then follows some simple $v_{k} \dsh t_i$ path is a simple $s_i\dsh t_i$ path.
Consequently, $P$ lies entirely in $B_k$.

\subsection{The complete proof of Proposition \ref{prop:hitVertexu}}\label{app:propHitVertexProof}

Consider an arbitrary simple $s_2\dsh t_2$ path $P$ that shares some edge with $B$. 
WLOG we may assume, again, that $P$ only visits vertices of $G_1$ and $s_2$, $t_2$, 
as we may handle subpaths of $P$ that have vertices that lie outside $G_1$ as edges. 
Let $B$ be the parallel combination of $H_1$ and $H_2$ and, assuming that the proposition does not hold, 
let $w$ be the first vertex of $P$ (before it reaches $u$) that belongs to the internal portion of $B$, 
and suppose WLOG that $w$ lies in $H_1$. 
By Lemma~\ref{lm:SePaProperties1}(ii) and Proposition \ref{prop:hitEdgeu}, the edge of $P$ exiting $w$ cannot go toward $t_1$, i.e.\ forward, 
as then it would traverse an edge of $B$ for the first time  that does not depart form $u$ (recall also Lemma~\ref{lm:SePaProperties2}(ii)).
Thus, it has to go toward $s_1$, i.e.\ backward, 
either staying in $H_1$ or going back to one of the preceding blocks of $B$. If by going backward it stays in $H_1$ then by the same argument it has again to move backward.
However, we can only repeat this process finitely often and eventually after possibly hitting some vertices of $H_1$ other than $w$ and  after possibly visiting vertices of blocks that precede $B$, $P$ hits $u$. Note that this happens  without $P$ having hit any vertex of $H_2$ prior to hitting $u$ (recall Lemma \ref{lm:SePaProperties2}(i)). 
 %
%
There are two cases.

The first case occurs when on traversing $P$ up to $u$,  
there is some edge of $B$ departing from $u$ that does not lead to $t_2$ 
without revisiting one of the previously visited vertices. 
Yet, since $P$ shares an edge with $B$,
by Proposition~\ref{prop:hitEdgeu}
there is an edge of $B$ that departs from $u$ and leads to $t_2$ 
without revisiting one of the previously visited vertices. 
Let $e_1$ be the first of the above edges and $e_2$ be the second one. 
For these edges, the same contradicting instance as in Proposition~\ref{prop:hitEdgeu} can be constructed. 
Path $P_1$ is an arbitrary simple $s_1\dsh t_1$ path containing $e_1$, 
path $P_3$ is an arbitrary simple $s_1 \dsh t_1$ path containing $e_2$,
and path $P_2$ is constructed by following $P$ up to $u$ 
(instead of some internal vertex of $H_1$ from which $e_2$ was departing) 
and from there taking $e_2$ and then a path that leads to $t_2$ without revisiting vertices. 
The edge function assignment will be exactly the same.  
Diversity hurting, and thus the contradiction, will follow in the same way as in Proposition~\ref{prop:hitEdgeu}.
 
The second and more interesting case occurs when 
on following $P$ up to $u$, all edges of $B$ departing from $u$ can lead to $t_2$ 
without revisiting one of the previously visited vertices. 
Let $e_2$ be an edge of $H_2$ (departing from $u$) with this property, 
let $P_2$ be the simple path that follows $P$ up to $u$ 
and then follows some path through $e_2$ to go to $t_2$,
and let $P_3$ be a simple $s_1\dsh t_1$ path that follows an arbitrary $s_1\dsh u$ path and 
an arbitrary $v\dsh s_2$ path,
and between $u$ and $v$ follows a path that contains $e_2$. 
Let $P_1$ be any path that follows an arbitrary $s_1\dsh u$ path and an arbitrary  $v\dsh t_2$ path,
and between $u$ and $v$ follows a path that contains $w$, and therefore goes through $H_1$. 
Let $e_1$ be the edge of $P_1$ that departs from $w$. 
Note that, because of Proposition~\ref{prop:hitEdgeu},  
$e_1$ cannot lead to $t_2$ with a simple path, i.e.\ without visiting vertices on $P$ before $w$. 
 This is a key fact for the contradiction to come. See Figure~\ref{fig:prop2App}.

To create the instance we proceed as in Proposition~\ref{prop:hitEdgeu}. 
Let  $d_1=d_2=1$ be the total demands for Commodities $1$ and $2$ respectively 
and let all other commodities have $0$ demand. 
Let $G_1$'s demand consist of $\tfrac 34$ \ntrl{ }players and $\tfrac 14$  players 
with \divpar{ }equal to $4$, and let the demand of $G_2$ consist of players 
with \divpar{ }equal to $1$.  
Assign edge $e_1$ the constant latency function $\ell_1(x)=1$ 
and the constant \stdev{ }function $\std_1(x)=2$. 
Assign edge $e_2$ the constant \stdev{ }$\std_2=0$,
and as latency function any $\ell_2$ that is continuous and strictly increasing 
with $\ell_2(1)=3$ and $\ell_2(\frac{5}{4})=9$. 
Assign all other edges the constant \stdev{ }$\std_2=0$.
To all other edges of $P_1$ and $P_3$ that lie inside $B$ assign $0$ latency functions. 
To all edges that depart from a vertex of $P_1$ or $P_3$ that lies on $P_2$ 
assign latency functions equal to some big constant $N$, say $N=24$
(i.e.\ double the heterogeneous equilibrium cost of Proposition~\ref{prop:JustTwoPaths}). 
For all remaining edges on $P_2$, assign $0$ latency functions. 
Finally, to all remaining edges, assign constant latency functions equal to $M$,
where $M$ is defined to be $2|V(G)|\cdot N$.

\begin{figure}\center
\includegraphics[scale=0.34]{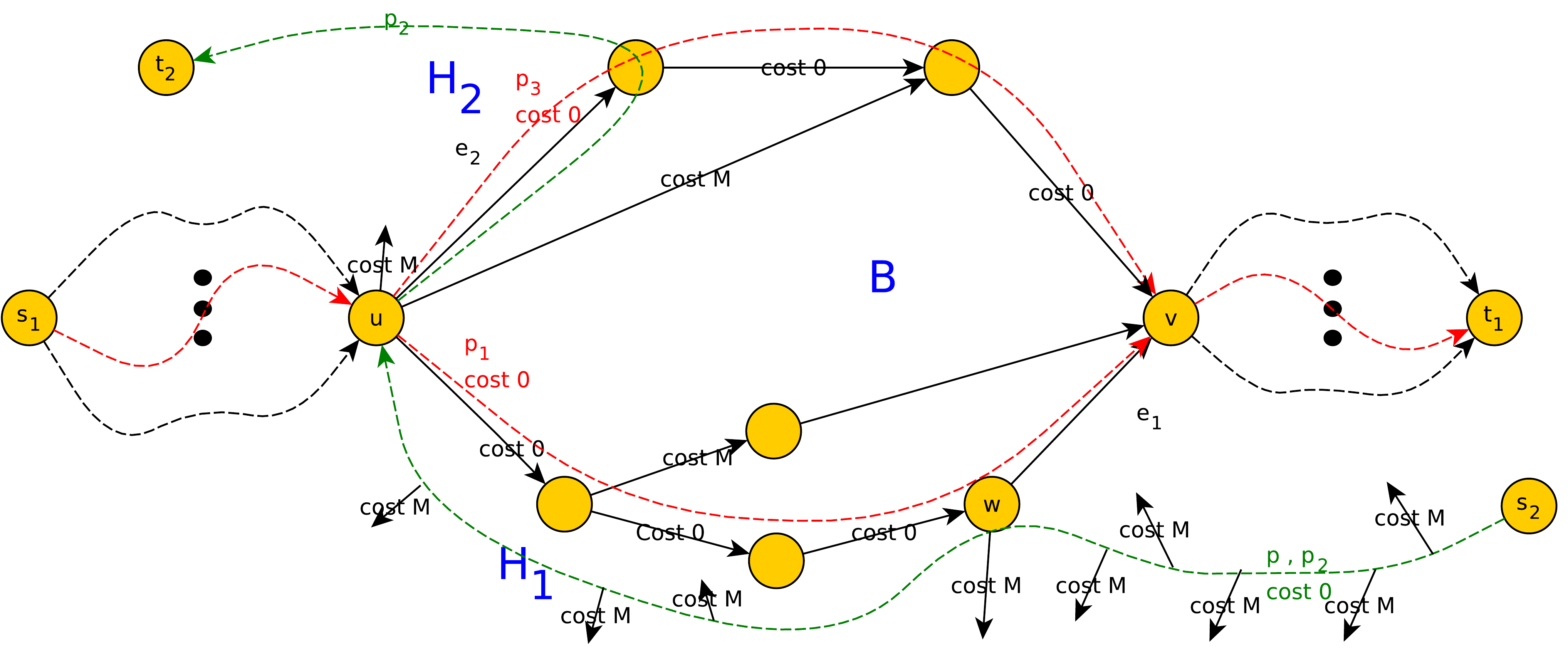} 
\caption {Sample network for Proposition \ref{prop:hitVertexu}}
\label{fig:prop2App}
\end{figure}
Note that (as in Proposition~\ref{prop:hitEdgeu}) all edges other than $e_2$ have constant edge functions. 
Thus for both equilibria, Commodity $1$ will have a common cost $C^{B^-}$ 
that will be paid on blocks other than $B$. 
Also, for Commodity $2$, any path that costs less than $M$ will follow path $P_2$ up to and including $e_2$. 
For this, it suffices to show that all edges departing from vertices of $P_1$ in between $u$ and $w$ 
have cost $M$, as then   if some portion of the flow, after visiting $u$, 
deviates and follows $P_1$ instead of $e_2$, then in order to avoid edges of cost $M$ 
it will reach $w$ which (as mentioned earlier) would be a dead end because of Proposition~\ref{prop:hitEdgeu}. This is proved in the next paragraph.
Given that, if 
both the homogeneous and the 
 heterogeneous equilibria are 
to cost less than $M$ 
then all of Commodity $2$'s flow will go through $e_2$
 and from there follow a shortest path to $t_2$
(recall that all other edges have constant edge functions) of cost $C^{e_2^-}$ say. Now note that $C^{e_2^-}$ will not be more than the cost of path $P_2$ following edge $e_2$,
which is bounded by
$|V(G)|N=\frac{M}{2}$,
and thus the path portion with cost  $C^{e_2^-}$ is preferable to any path with an edge of cost $M$.

\hide{
Note that $C^{e_2^-}$ cannot be more than $|E(G)|N=\frac{M}{2}$. 
This follows by the fact that edges with cost $M$ cannot be visited 
and the maximum cost among all other edges is at most $N$.\footnotemark[\getrefnumber{footnt:EdgesUpperCost}] The fact that edges with cost $M$ 
cannot be visited  
is partially true
because of the simplicity of $P_2$ (which was the argument in Proposition~\ref{prop:hitEdgeu}), 
but in this case an extra argument is needed for the edges that depart from vertices of $P_1$ 
in between $u$ and $w$. 
We need a guarantee that these edges do not belong in $P_2$, i.e.\ after $P_2$ visits $u$ and $e_2$, 
it does not go (at any later point) to one of these vertices to take one of their departing edges. }
%
 
 Let $E_M$ be the set of  edges that depart from vertices of $P_1$ in between $u$ and $w$. 
  We want to prove that edges in $E_M$ have cost $M$. 
By proposition~\ref{prop:hitEdgeu}, the vertices of $P$ that belong to $P_1$ in between $w$ and $u$, 
$w$  included, 
have no departing edge that belongs to $B$ and leads to $t_2$ without traversing preceding vertices of $P$. 
This implies that $P$ and any other simple $s_2\dsh t_2$ path that follows $P$ up to $u$, 
cannot have some simple $u\dsh v$ path of $H_2$ as a subpath --- call this Property $X$ --- 
otherwise, by letting $P'$ be such a path,  following $P$ up to $w$, 
then picking any path inside $H_1$ that leads to $v$, and from there reaching $t_2$ via $P'$, 
creates a simple $s_2\dsh t_2$ path that has its first edge in $B$ departing from an internal vertex, thereby contradicting Proposition~\ref{prop:hitEdgeu}.  
But if $P_2$ is to contain some edge in $E_M$ then it has to leave $H_2$ and go to $H_1$. 
The only way to do that and keep its simplicity, because of  
Lemmas~\ref{lm:SePaProperties1}(ii) and~\ref{lm:SePaProperties2}(i), 
is by going to a block that succeeds $B$ and then coming back to $H_1$. 
But  by Lemma~\ref{lm:SePaProperties2}(ii),
going to a block that succeeds $B$ requires 
going through $v$ first.
Thus $P_2$ would have a complete $u\dsh v$ subpath that does not visit any other block, which 
by Lemma~\ref{lm:SePaProperties2}(iii) belongs in $B$ and thus in $H_2$, 
contradicting Property $X$. Thus $P_2$ does not contain any edge in $E_M$.


Now, we 
compute the costs of the equilibria 
The heterogeneous equilibrium $g$, 
for Commodity $1$, routes all the flow through the shortest $s_1\dsh u$ and $v\dsh t_1$ paths,
inside $B$ routes $\tfrac 34$ units of flow through $P_1$, and  $\tfrac 14$ units of flow through $P_2$, 
and routes all the flow of Commodity $2$ through $e_2$, via $P_2$ up to $e_2$ 
and after $e_2$ 
via the shortest path to $t_2$. 
The \ntrl{ }players of $G_1$ compute a cost for $P_1$ equal to $1+C^{B^-}$ 
and a cost for $P_3$ equal to $9+C^{B^-}$, and thus prefer $P_1$ to $P_3$,
while the remaining players of $G_1$ compute a cost equal to $9+C^{B^-}$ for both $P_1$ and $P_3$,
and thus prefer $P_3$ to $P_1$ (recall that $\ell_2$ is strictly increasing). 
The other paths have cost at least $N+C^{B^-}=24+C^{B^-}$ and thus are not preferred by any type of player. 
The players of Commodity $2$ pay cost equal to $9+C^{e_2^-}$ 
(the cost of $e_2$ plus the cost after it) and thus prefer staying on $e_2$ 
rather than paying at least $M$ to avoid $e_2$ 
(recall $C^{e_2^-}\leq \frac{M}{2}$). 
Also recall that on $P_2$, for the vertices before $e_2$, there might be edges leaving $P_2$ that have
cost $0$ 
(these are edges of $P_1$
),
but because of  Proposition~\ref{prop:hitEdgeu} 
they cannot lead to $t_2$ without visiting preceding vertices. 
Putting it all together this routing is indeed the heterogeneous equilibrium with cost 
$C^{ht}(g)=1\cdot \tfrac 34 d_1+9\cdot \tfrac 14 d_1+d_1C^{B^-}+9d_2+d_2C^{e_2^-}=12+C^{B^-}+C^{e_2^-}$.

The homogeneous equilibrium $f$ routes all the demand of $G_1$ through  $e_1$, 
via $P_1$ and the shortest $s_1\dsh u$ and $v\dsh t_1$ paths, 
and routes all the flow of Commodity $2$ through $e_2$, via $P_2$ up to $e_2$,
and after $e_2$ 
 via the shortest path to $t_2$. 
The average \divpar{ }for Commodity $1$'s demand equals $1$,
and thus $P_1$ and $P_3$ are both computed to cost $3+C^{B^-}$, 
while all other paths cost at least $N+C^{B^-}=24+C^{B^-}$ and thus are avoided. 
In the same way as above, the players of Commodity $2$ pay cost equal to $3+C^{e_2^-}$ 
(the cost of $e_2$ plus the cost after it) and thus prefer staying on $e_2$ rather than 
paying at least $M$ to avoid it. 
Thus the cost $C^{hm}(f)$ of the homogeneous equilibrium is 
$C^{hm}(f)=3\cdot d_1+d_1\cdot  C^{B^-}+3\cdot d_2+d_2\cdot C^{e_2^-}=6+C^{B^-}+C^{e_2^-}$. 
Consequently, $C^{ht}(g)>C^{hm}(f)$,  
contradicting the assumption that $G$ satisfies that under any demand and edge functions diversity helps. Consequently, the proposition holds. \hfill\qed

\end{document}